\documentclass[]{article}
\usepackage[english]{babel}
\usepackage{tikz}
\usepackage{amssymb}
\usepackage[mathscr]{eucal}
\usepackage[T1]{fontenc}
\usepackage[utf8]{inputenc}
\usepackage{calrsfs}
\usepackage{float}
\usepackage{amsmath}
\usepackage{amsthm}
\usepackage{enumerate}
\usepackage{caption}
\usepackage{multirow}
\usepackage[titletoc,toc]{appendix}
\usepackage{authblk}
\usetikzlibrary[shapes.misc]
\usetikzlibrary[shapes.geometric]
\newtheorem{de}{Definition}[section]
\newtheorem{theo}{Theorem}
\newtheorem{prop}[de]{Proposition}
\newtheorem{cor}[de]{Corollary}

\newtheorem{obs}[de]{Observation}

\title{Almost disjoint spanning trees: relaxing the conditions for completely independent spanning trees}

\author[1]{Benoit Darties}
\author[2]{Nicolas Gastineau}
\author[1]{Olivier Togni}
\affil[1]{\textit{Universit\'e de Bourgogne, 21078 Dijon cedex, France, Le2i, UMR CNRS 6303} }
\affil[2]{\textit{PSL, Université Paris-Dauphine, LAMSADE UMR CNRS 7243, France} }

\begin{document}
\maketitle
\begin{abstract}
The search of spanning trees with interesting disjunction properties has led to the introduction of edge-disjoint spanning trees,  independent spanning trees and more recently completely independent spanning trees. We group together these notions by defining $(i,j)$-disjoint spanning trees, where $i$ ($j$, respectively) is the number of vertices (edges, respectively) that are shared by more than one tree. We illustrate how $(i,j)$-disjoint spanning trees provide some nuances between the existence of disjoint connected dominating sets and completely independent spanning trees. We prove that determining if there exist two $(i,j)$-disjoint spanning trees in a graph $G$ is NP-complete, for every two positive integers $i$ and $j$. Moreover we prove that for square of graphs, $k$-connected interval graphs, complete graphs and several grids, there exist $(i,j)$-disjoint spanning trees for interesting values of $i$ and $j$.
\end{abstract}
\section{Introduction}
The graphs considered are assumed to be connected, since spanning trees are only interesting for connected graphs.
Let $k\ge 2$ be an integer and $T_1,\ldots, T_k$ be \emph{spanning trees} in a graph $G$. The spanning trees $T_1,\ldots, T_k$ are {\em edge-disjoint} if $\cup_{1\le \ell<\ell' \le k}E(T_\ell)\cap E(T_{\ell'})=\emptyset$.
A vertex is said to be an \emph{inner vertex} in a tree $T$ if it has degree at least 2 in $T$ and a \emph{leaf} if it has degree 1. We denote by $I(T)$ the set of inner vertices of tree $T$.
The spanning trees $T_1,\ldots, T_k$ are \emph{internally vertex-disjoint} if $I(T_1), \ldots, I(T_k)$ are pairwise disjoint.
Finally, the spanning trees $T_1,\ldots, T_k$ are \emph{completely independent spanning trees} if they are both pairwise edge-disjoint and internally vertex-disjoint.

In this paper, we introduce $(i,j)$-disjoint spanning trees:

\begin{de}
Let $k\ge 2$ be an integer and $T_1,\ldots, T_k$ be spanning trees in a graph $G$.
We let $I(T_1,\ldots, T_k)=\{u\in V(G)|\exists \ell , \ell' \  u\in I(T_{\ell} ) \cap I(T_{\ell '}), \ 1\le \ell <\ell' \le k \}$ be the set of vertices which are inner vertices in at least two spanning trees among $T_1,\ldots, T_k$, and we let  $E(T_1,\ldots, T_k)=\{e\in E(G)|\exists \ell , \ell', \ 1\le \ell <\ell' \le k, \  e\in E(T_{\ell} ) \cap E(T_{\ell '}) \}$  be the set of edges which belong to at least two spanning trees among $T_1,\ldots, T_k$. The spanning trees $T_1,\ldots, T_k$ are \emph{$(i,j)$-disjoint} for two positive integers $i$ and $j$, if the two following conditions are satisfied:
\begin{enumerate}
\item[i)] $|I(T_1,\ldots, T_k)|\le i$;
\item[ii)] $|E(T_1,\ldots, T_k)| \le j$.
\end{enumerate}
\end{de}
By $*$ we denote a large enough integer, i.e. an integer larger than $\max(|E(G)|, |V(G)|)$, for a graph $G$. Remark that $(0,0)$-disjoint spanning trees are completely independent spanning trees and that $(*,0)$-disjoint spanning trees are edge-disjoint spanning trees.
Notice also that there are infinitely many $(i,j)$-disjoint trees in $G$, for $i\ge \gamma_c(G)$ and $j\ge |V(G)|-1$, $\gamma_c(G)$ being the minimum size of a connected dominating set in $G$ (one can repeat infinitely the same tree with $\gamma_c(G)$ inner vertices).
\subsection{Related work}

Completely independent spanning trees were introduced by Hasunuma \cite{HA2001} and then have been studied on different classes of graphs, such as underlying graphs of line graphs \cite{HA2001}, maximal planar graphs \cite{HA2002}, Cartesian product of two cycles \cite{HA2012}, complete graphs, complete bipartite and tripartite graphs \cite{PAI2013}, variant of hypercubes \cite{CH2017,PAI2016} and chodal rings \cite{PAI2014}.
Moreover, determining if there exist two completely independent spanning trees in a graph $G$ is a NP-hard problem \cite{HA2002}. Recently, sufficient conditions inspired by the sufficient conditions for hamiltonicity have been determined in order to guarantee the existence of two completely independent spanning trees: Dirac's condition \cite{AR2014} and Ore's condition \cite{FA2014}. Moreover, Dirac's condition has been generalized to more than two trees \cite{CH2015,HA2016,HO2016} and has been independently improved \cite{HA2016,HO2016} for two trees. Also, a recent paper has studied the problem on the class of $k$-trees, for which the authors have proven that there exist at least $\lceil k/2 \rceil$ completely independent spanning trees \cite{MA2015}.

For a given tree $T$ and a given pair of vertices $(u,v)$ of $T$, let $P_{T}(u,v)$ be the set of vertices in the unique path between $u$ and $v$ in $T$. 
Remark that $T_1,\ldots, T_k$ are internally vertex-disjoint in a graph $G$ if and only if for any pair of vertices $(u,v)$ of $V(G)$, $\cup_{1\le \ell <\ell ' \le k}P_{T_\ell}(u,v)\cap P_{T_{\ell'}}(u,v)=\{u,v\}$.
Other works on disjoint spanning trees include independent spanning trees, i.e. focus on finding spanning trees $T_1,\ldots, T_k$ rooted at the same vertex $r$. In independent spanning trees, for any vertex $v$ the paths between $r$ and $v$ in $T_1,\ldots, T_k$ are pairwise internally vertex-disjoint, i.e. for each integers $i$ and $j$, $1\le i<j\le k$, $P_{T_i}(r,v)\cap P_{T_j}(r,v)=\{r,v\}$. In contrast with the notion of completely independent spanning trees, in independent spanning trees only the paths to $r$ are considered. Thus, $T_1,\ldots, T_k$ may share common vertices or edges, which is not admissible with completely independent spanning trees.  Independent spanning trees have been studied for several classes of graphs which include product graphs \cite{OIBI1996}, de Bruijn and Kautz digraphs \cite{GH1997,HN2001}, and chordal rings \cite{IKOI1999}. Related works also include edge-disjoint spanning trees, i.e. spanning trees which are pairwise edge-disjoint only. Edge-disjoint spanning trees have been studied on many classes of graphs, including hypercubes \cite{BA1999}, Cartesian product of cycles \cite{WA2001} and Cartesian product of two graphs \cite{KU2003}.

Some subsets of vertices $D_1,\ldots,D_k$ of a graph $G$ are $k$ \emph{disjoint connected dominating sets} if $D_1,\ldots,D_k$ are pairwise disjoint and each subset is a connected dominating set in $G$. There are some works about disjoint connected dominating sets that can be transcribed in terms of internally vertex-disjoint spanning trees (the disjoint connected dominating sets can be used to provide the inner vertices of internally vertex-disjoint spanning trees). The maximum number of disjoint connected dominating sets in a graph $G$ is the \emph{connected domatic number}. This parameter is denoted by $d_c(G)$ and has been introduced by Hedetniemi and Laskar \cite{HE1984} in 1984.  An interesting result about connected domatic number concerns planar graphs, for which Hartnell and Rall have proven that, except $K_4$ (which has connected domatic number $4$), their connected domatic number is bounded by 3 \cite{HAR2001}. The problem of constructing a connected dominating set is often motivated by wireless ad-hoc networks \cite{GU1998,WA2004} for which connected dominating sets are used to create a virtual backbone in the network.

\subsection{Motivation and basic facts about disjoint dominating sets}
Remark that $(0,*)$-disjoint spanning trees are internally vertex-disjoint, and consequently, are related to connected dominating sets. Hence, we call $(0,*)$-disjoint spanning trees, \emph{trees induced by disjoint connected dominating sets} and we give the properties about $(0,*)$-disjoint spanning trees using, when possible, the concept of disjoint connected dominating sets. Figure \ref{desCDS} illustrates how disjoint connected dominating sets are used to construct  $(0,*)$-disjoint spanning trees. As we observe in the next proposition, trees induced by disjoint connected dominating sets satisfy interesting properties.
First, an edge can only belong to at most two trees (Proposition \ref{megaprop}.i)). Second, the paths between two non-adjacent vertices in trees induced by disjoint connected dominating sets are edge-disjoint (Proposition \ref{megaprop}.ii)). Moreover, the fact that the paths between two adjacent vertices share a common edge implies that these vertices are inner vertices in different trees (Proposition \ref{megaprop}.iii)).
These properties illustrate the utility of disjoint connected dominating sets to broadcast a message following multiples routes in a network. 
For a spanning tree, an \emph{inner edge} is an edge between two inner vertices and a \emph{leaf edge} is an edge which is not an inner edge.
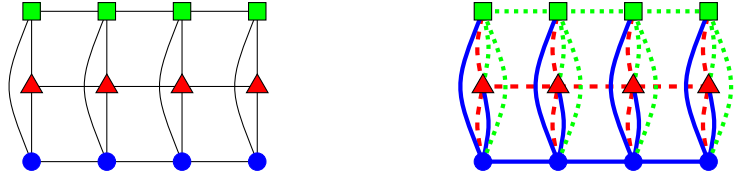
\begin{figure}[t]
\begin{center}
\begin{tikzpicture}
\draw (0,0) -- (3,0);
\draw (0,1) -- (3,1);
\draw (0,2) -- (3,2);
\draw (0,0) -- (0,2);
\draw (1,0) -- (1,2);
\draw (2,0) -- (2,2);
\draw (3,0) -- (3,2);
\draw (0,0) .. controls (-0.4,1).. (0,2);
\draw (1,0) .. controls (0.6,1).. (1,2);
\draw (2,0) .. controls (1.6,1).. (2,2);
\draw (3,0) .. controls (2.6,1).. (3,2);

\node at (0,0) [circle,draw=blue,fill=blue,scale=0.7]{};
\node at (0,1) [regular polygon, regular polygon sides=3,draw=black,fill=red,scale=0.5]{};
\node at (0,2) [regular polygon, regular polygon sides=4,draw=black,fill=green,scale=0.7]{};
\node at (1,0) [circle,draw=blue,fill=blue,scale=0.7]{};
\node at (1,1) [regular polygon, regular polygon sides=3,draw=black,fill=red,scale=0.5]{};
\node at (1,2) [regular polygon, regular polygon sides=4,draw=black,fill=green,scale=0.7]{};
\node at (2,0) [circle,draw=blue,fill=blue,scale=0.7]{};
\node at (2,1) [regular polygon, regular polygon sides=3,draw=black,fill=red,scale=0.5]{};
\node at (2,2) [regular polygon, regular polygon sides=4,draw=black,fill=green,scale=0.7]{};
\node at (3,0) [circle,draw=blue,fill=blue,scale=0.7]{};
\node at (3,1) [regular polygon, regular polygon sides=3,draw=black,fill=red,scale=0.5]{};
\node at (3,2) [regular polygon, regular polygon sides=4,draw=black,fill=green,scale=0.7]{};

\draw[ultra thick, style=dashed,color=red] (0+6,0) .. controls (-0.1+6,0.5).. (0+6,1);
\draw[ultra thick, style=dashed,color=red] (1+6,0) .. controls (0.9+6,0.5).. (1+6,1);
\draw[ultra thick, style=dashed,color=red] (2+6,0) .. controls (1.9+6,0.5).. (2+6,1);
\draw[ultra thick, style=dashed,color=red] (3+6,0) .. controls (2.9+6,0.5).. (3+6,1);
\draw[ultra thick, style=dashed,color=red] (0+6,1) .. controls (-0.1+6,1.5).. (0+6,2);
\draw[ultra thick, style=dashed,color=red] (1+6,1) .. controls (0.9+6,1.5).. (1+6,2);
\draw[ultra thick, style=dashed,color=red] (2+6,1) .. controls (1.9+6,1.5).. (2+6,2);
\draw[ultra thick, style=dashed,color=red] (3+6,1) .. controls (2.9+6,1.5).. (3+6,2);
\draw[ultra thick, style=dashed,color=red] (3+6,1) -- (0+6,1);
\draw[ultra thick,color=blue] (0+6,0) .. controls (0.1+6,0.5).. (0+6,1);
\draw[ultra thick,color=blue] (1+6,0) .. controls (1.1+6,0.5).. (1+6,1);
\draw[ultra thick,color=blue] (2+6,0) .. controls (2.1+6,0.5).. (2+6,1);
\draw[ultra thick,color=blue] (3+6,0) .. controls (3.1+6,0.5).. (3+6,1);
\draw[ultra thick,color=blue] (0+6,0) .. controls (-0.4+6,1).. (0+6,2);
\draw[ultra thick,color=blue] (1+6,0) .. controls (0.6+6,1).. (1+6,2);
\draw[ultra thick,color=blue] (2+6,0) .. controls (1.6+6,1).. (2+6,2);
\draw[ultra thick,color=blue] (3+6,0) .. controls (2.6+6,1).. (3+6,2);
\draw[ultra thick,color=blue] (0+6,0) -- (3+6,0);
\draw[ultra thick, style=dotted,color=green] (3+6,2) -- (0+6,2);
\draw[ultra thick, style=dotted,color=green] (0+6,1) .. controls (0.1+6,1.5).. (0+6,2);
\draw[ultra thick, style=dotted,color=green] (1+6,1) .. controls (1.1+6,1.5).. (1+6,2);
\draw[ultra thick, style=dotted,color=green] (2+6,1) .. controls (2.1+6,1.5).. (2+6,2);
\draw[ultra thick, style=dotted,color=green] (3+6,1) .. controls (3.1+6,1.5).. (3+6,2);
\draw[ultra thick, style=dotted,color=green] (0+6,0) .. controls (0.4+6,1).. (0+6,2);
\draw[ultra thick, style=dotted,color=green] (1+6,0) .. controls (1.4+6,1).. (1+6,2);
\draw[ultra thick, style=dotted,color=green] (2+6,0) .. controls (2.4+6,1).. (2+6,2);
\draw[ultra thick, style=dotted,color=green] (3+6,0) .. controls (3.4+6,1).. (3+6,2);

\node at (0+6,0) [circle,draw=blue,fill=blue,scale=0.7]{};
\node at (0+6,1) [regular polygon, regular polygon sides=3,draw=black,fill=red,scale=0.5]{};
\node at (0+6,2) [regular polygon, regular polygon sides=4,draw=black,fill=green,scale=0.7]{};
\node at (1+6,0) [circle,draw=blue,fill=blue,scale=0.7]{};
\node at (1+6,1) [regular polygon, regular polygon sides=3,draw=black,fill=red,scale=0.5]{};
\node at (1+6,2) [regular polygon, regular polygon sides=4,draw=black,fill=green,scale=0.7]{};
\node at (2+6,0) [circle,draw=blue,fill=blue,scale=0.7]{};
\node at (2+6,1) [regular polygon, regular polygon sides=3,draw=black,fill=red,scale=0.5]{};
\node at (2+6,2) [regular polygon, regular polygon sides=4,draw=black,fill=green,scale=0.7]{};
\node at (3+6,0) [circle,draw=blue,fill=blue,scale=0.7]{};
\node at (3+6,1) [regular polygon, regular polygon sides=3,draw=black,fill=red,scale=0.5]{};
\node at (3+6,2) [regular polygon, regular polygon sides=4,draw=black,fill=green,scale=0.7]{};
\end{tikzpicture}
\end{center}
\caption{Three disjoint connected dominating sets in $C_3\square P_4$ (on the left) and the spanning trees induced by these dominating sets (on the right) (circles: $D_1$ or $I(T_1)$; triangles: $D_2$ or $I(T_2)$; squares: $D_3$ or $I(T_3)$; plain lines: edges of $T_1$; dashed lines: edges of $T_2$; dotted lines: edges of $T_3$).}
\label{desCDS}
\end{figure}
\begin{prop}\label{megaprop}
Let $i$ and $j$ be two integers, $1\le i<j \le k$.
Let $G$ be a graph of order at least $3$, let $T_1, \ldots, T_k$ be spanning trees induced by $k$ disjoint connected dominating sets and let $u,v\in V(G)$. 
\begin{enumerate}
\item[i)] every edge belongs to at most two trees among $T_1, \ldots, T_k$;
\item[ii)] if $u$ and $v$ are not adjacent, then $P_{T_i}(u,v) \cap P_{T_j} (u,v) =\emptyset$;
\item[iii)] if $P_{T_i}(u,v) \cap P_{T_j} (u,v)\neq\emptyset$ then $\{u,v\}\not\subseteq I(T_i)$ and $\{u,v\}\not\subseteq I(T_j)$.
\end{enumerate} 
\end{prop}
\begin{proof}
We prove that each of the three properties holds.\newline
i) Suppose that $uv$ is an inner edge in a spanning tree. Since the vertices $u$ and $v$ are leaves in any other tree, $uv$ cannot belong to more than one spanning tree. Suppose $uv$ is a leaf edge in at least two trees. The edge $uv$ can belong to at most two trees, the trees for which $u$ and $v$ are inner vertices.\newline
ii) Since the paths between $u$ and $v$ in the different trees have length at least $2$ and contain no common inner vertices, they share no common edges.\newline
iii) By Property ii), $u$ and $v$ are adjacent. Moreover, if $u,v\in I(T_i)$, then $P_{T_i} (u,v)$ only contains inner edges of $T_i$, and, as for Property i), each inner edge can not belong to another tree. Since $P_{T_i} (u,v)$ only contains inner edges of $T_i$, $P_{T_i}(u,v)\cap P_{T_j} (u,v)=\emptyset$. The same goes if $u,v\in I(T_j)$.
\end{proof}
Note that there is a relation between the minimum size of a connected dominating set in a graph $G$, denoted by $\gamma_c (G)$ and $d_c(G)$ (the maximum number of disjoint connected dominating sets) since $d_c(G)\le \lfloor |V(G)|/ \gamma_c(G) \rfloor$.
We also have to mention that Fan, Hong and Liu \cite{FA2014} have studied the line graph of cubic graphs of order at least $10$ and have proven that there are no two completely independent spanning trees in these cubic graphs. It could be possible, however, that it is not the case for two disjoint dominating sets.

If a graph satisfies $d_c(G)=k$ and does not contain $k$ completely independent spanning trees, then there exist an integer $j$ such that $G$ contains $k$ $(0,j)$-disjoint spanning trees.
Hence, the notion of $(0,j)$-disjoint spanning trees provides some nuances between the existence of disjoint connected dominating sets and completely independent spanning trees.

We say that $k$ connected dominating sets $D_1$, $\ldots$, $D_k$ are \emph{$\ell$-rooted connected dominating sets} if the set $A=\cup_{1\le i< j \le k} D_i\cap D_j $ satisfies $|A|\le \ell$. Remark that we can construct $k$ $(\ell,*)$-disjoint spanning trees in a graph that contains $k$ $\ell$-rooted connected dominating sets $D_1$, $\ldots$, $D_k$ by considering that $I(T_i)=D_i$, for every integer $i$, $1\le i\le k$.
Note also that trees $T_1,\ldots,T_k$ induced by $1$-rooted connected sets, i.e. $(1,*)$-disjoint spanning trees, are also independent spanning trees rooted at a vertex $r\in I(T_1,\ldots,T_k)$.
However, if $T_1,\ldots, T_k$ are independent spanning trees rooted at $r$ in $G$, then $T_1,\ldots, T_k$ are not always $(1,*)$-disjoint spanning trees in $G$. This difference is illustrated by the fact that if for two vertices $u,v\in V(G)$ and two spanning trees $T_i$ and $T_j$, $i\neq j$ , we have $P_{T_i}(u,r) \cap P_{T_j} (u,r) =\{ u,r\}$ and $P_{T_i}(v,r) \cap P_{T_j} (v,r) =\{ v,r\}$, then it does not imply that $P_{T_i}(u,v) \cap P_{T_j} (u,v) =\{ u,v,r\}$.

\subsection{Notation and Organization}
We denote by $\delta(G)$ the minimum degree of $G$, i.e., $\delta(G)=\min\{ N(u)|\ u \in V(G)\}$.
We denote by $d_G(u,v)$ the usual distance between two vertices $u$ and $v$ in a graph $G$.
The graph $G-e$ is the graph obtained from $G$ by removing an edge $e$ from $E(G)$ and $G-A$, for $A\subseteq V(G)$, is the graph obtained from $G$ by removing the vertices from $A$ and their incident edges.
For $A\subseteq V(G)$, we denote by $G[A]$, the graph $G-(V(G)\setminus A)$.
We say that a graph $G$ is $k$-connected if $|V(G)|\ge k+1$ and if for any set of vertices $A\subseteq V(G)$, with $|A|\le k-1$, $G-A$ is connected.
By $K_{n}$, $P_n$ and $C_n$, we denote the complete graph, path and cycle, respectively, of order $n$.
Let $n_1$ and $n_2$ be positive integers. By $G(n_1,n_2)$ we denote the \emph{square grid} with $n_1$ rows and $n_2$ columns. The graph $G(n_1,n_2)$ can be also defined as the Cartesian product of two paths $P_{n_1}$ and $P_{n_2}$. The \emph{cylinder}, denoted by $C(n_1,n_2)$, is the Cartesian product of one cycle $C_{n_1}$ and one path $P_{n_2}$.

This article is organized as follows. Section 2 presents alternative characterizations of $(i,j)$-disjoint spanning trees.
Section 3 is about the computational complexity of the following decision problem: is it true that a graph $G$ contains two $(i,j)$-disjoint spanning trees (with input the graph $G$).
Section 4 deals with $k$-connectivity and the conditions of Dirac and Ore for $(i,j)$-disjoint spanning trees.
Section 5 is about the required number of edges and distribution of inner vertices in $(i,j)$-disjoint spanning trees.
Section 6 presents some $(i,j)$-disjoint spanning trees in square of graphs, $k$-connected interval graphs, complete graphs, and square grids and cylinders.

\section{Characterizations in terms of partitions and dominating sets}
We begin this section by proving the following proposition.
\begin{prop}\label{prop1}
Let $G$ be a connected graph of order at least $3$ and let $T_1, \ldots, T_k$ be $(i,j)$-disjoint spanning trees in $G$. For every integer $\ell$, $1\le \ell \le k$, every vertex $u\in V(G)$ satisfies the two following properties:
\begin{enumerate}
\item[i)] if $u\notin I(T_\ell)$, then $u$ has a neighbor in $I(T_\ell)$;
\item[ii)] if $G$ has diameter at least $3$, then $u$ has a neighbor in $I(T_\ell)$.
\end{enumerate} 
\end{prop}
\begin{proof}
Suppose there exist an integer $\ell$, $1\le \ell \le k$, and a vertex $u$ which has no neighbor in $I(T_\ell)$.\newline
i) If $u\notin I(T_\ell)$, then $G=T_\ell=P_2$ which contradicts the hypothesis that $G$ has order at least $3$.
\newline
ii) Since property i) holds, we suppose that $u\in I(T_\ell)$. Remark that since $G$ has diameter at least $3$, a spanning tree of $G$ has also diameter at least $3$.
Moreover, if $u$ is only adjacent to leaf vertices then it implies that $T_\ell$ is a star which contradicts the fact that $T_\ell$ has diameter at least $3$.
\end{proof}
Let $V_1$ and $V_2$ be two subsets of vertices of a graph $G$. By $B(V_1,V_2)$ we denote the bipartite graph with vertex set $V_1\cup V_2$ and edge set $\{uv\in E(G) |\ u\in V_1,\ v\in V_2\}$.
In the two following subsections we give alternative characterizations of $(0,\ell)$-disjoint spanning trees and $\ell$-rooted connected dominating sets. 
These characterizations are expressed in terms of partition in sets of vertices fulfilling some properties.

\subsection{$(0,\ell)$-disjoint spanning trees}
In this subsection, we introduce a definition which is inspired by the definition of CIST-partition introduced by Araki \cite{AR2014}.
\begin{de}
An \emph{$\ell$-CIST-partition} of a graph $G$ into $k$ sets is a partition of $V(G)$ into $k$ sets of vertices $V_1,\ldots,V_k$ such that:
\begin{enumerate}
\item[i)] $G[V_i]$ is connected, for each integer $i$, $1\le i\le k$;
\item[ii)] $B(V_i,V_j)$ contains no isolated vertex, for every two integers $i$, $j$, $1\le i<j\le k$;
 \item[iii)] $\sum_{1\le i<j\le k} c_{i,j}\le \ell$, where $c_{i,j}$ is the number of connected component which are trees in $B(V_i,V_j)$, $1\le i<j\le k$. 
\end{enumerate}
\end{de}
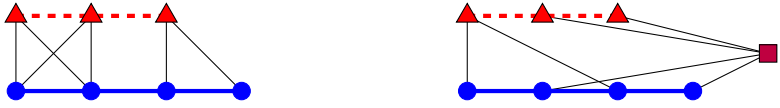
\begin{figure}[t]
\begin{center}
\begin{tikzpicture}
\draw (0,0) -- (0,1);
\draw (0,0) -- (1,1);
\draw (1,0) -- (0,1);
\draw (1,0) -- (1,1);
\draw (2,0) -- (2,1);
\draw (3,0) -- (2,1);
\draw[ultra thick,color=blue] (0,0) -- (3,0);
\draw[ultra thick, style=dashed,color=red] (0,1) -- (2,1);
\node at (0,0) [circle,draw=blue,fill=blue,scale=0.7]{};
\node at (0,1) [regular polygon, regular polygon sides=3,draw=black,fill=red,scale=0.5]{};
\node at (1,0) [circle,draw=blue,fill=blue,scale=0.7]{};
\node at (1,1) [regular polygon, regular polygon sides=3,draw=black,fill=red,scale=0.5]{};
\node at (2,0) [circle,draw=blue,fill=blue,scale=0.7]{};
\node at (2,1) [regular polygon, regular polygon sides=3,draw=black,fill=red,scale=0.5]{};
\node at (3,0) [circle,draw=blue,fill=blue,scale=0.7]{};

\draw (0+6,0) -- (0+6,1);
\draw (2+6,0) -- (0+6,1);
\draw (1+6,0) -- (4+6,0.5);
\draw (3+6,0) -- (4+6,0.5);
\draw (1+6,1) -- (4+6,0.5);
\draw (2+6,1) -- (4+6,0.5);
\draw[ultra thick,color=blue] (0+6,0) -- (3+6,0);
\draw[ultra thick, style=dashed,color=red] (0+6,1) -- (2+6,1);
\node at (0+6,0) [circle,draw=blue,fill=blue,scale=0.7]{};
\node at (0+6,1) [regular polygon, regular polygon sides=3,draw=black,fill=red,scale=0.5]{};
\node at (1+6,0) [circle,draw=blue,fill=blue,scale=0.7]{};
\node at (1+6,1) [regular polygon, regular polygon sides=3,draw=black,fill=red,scale=0.5]{};
\node at (2+6,0) [circle,draw=blue,fill=blue,scale=0.7]{};
\node at (2+6,1) [regular polygon, regular polygon sides=3,draw=black,fill=red,scale=0.5]{};
\node at (3+6,0) [circle,draw=blue,fill=blue,scale=0.7]{};
\node at (4+6,0.5) [regular polygon, regular polygon sides=4,draw=black,fill=purple,scale=0.7]{};
\end{tikzpicture}
\end{center}
\caption{An $1$-CIST partition (on the left) and an $1$-rooted partition (on the right) of two graphs (circles: $V_1$; triangles: $V_2$; square: $A$).}
\label{despart}
\end{figure}

Figure \ref{despart} illustrates an $1$-CIST partition on a specific graph.
In a similar way than Araki \cite{AR2014}, we prove that the notions of $\ell$-CIST partition and $(0,\ell)$-disjoint spanning trees are equivalent.
\begin{theo}\label{theopart1}
Let $G$ be a graph. There exist $k$ $(0,\ell)$-disjoint spanning trees $T_1, \ldots, T_k$ in $G$ if and only if $G$ has an $\ell$-CIST-partition into $k$ sets.
\end{theo}
\begin{proof}
Suppose $G$ has an $\ell$-CIST-partition into $k$ sets $V_1$,\ldots,$V_k$. We are going to construct $(0,\ell)$-disjoint spanning trees $T_1, \ldots, T_k$. We begin by setting $I(T_i)= V_i$ for each integer $i$, $1\le i\le k$. For each integer $i$, $1\le i\le k$, we suppose that $E(T_i)$ is empty and we progressively add edges to $T_i$ in order to obtain spanning trees of $G$ at the end of the proof. Since $G[V_i]$ is connected for each integer $i$, $1\le i \le k$, it is possible to add edges to $T_i$ in order to have a spanning tree with inner vertices from $V_i$, for each integer $i$. 
%Fix two integers $i$ and $j$, $1\le i<j\le k$, and let $D_{i,j}$ be a connected component of $B(V_i, V_j)$. We will prove that we can add edges in both $E(T_i)$ and $E(T_j)$ in order to obtain a spanning tree restricted to $V_i \cup V(D_{i,j})$ and another spanning tree restricted to $V_j \cup V(D_{i,j})$. 

Let $i$ and $j$ be two integers, $1\le i<j\le k$, and let $D_{i,j}$ be a connected component of $B(V_i, V_j)$. We add edges in order to build a spanning tree restricted to $V_i \cup V(D_{i,j})$ and another spanning tree restricted to $V_j \cup V(D_{i,j})$ by considering two cases. Let $u$ be a vertex of $D_{i,j}\cap V_i$. First, if $D_{i,j}$ is a tree, then we add an edge $e$ of $D_{i,j}$ incident with $u$ to both $T_i$ and $T_j$. Thus, the edge $e$ will be common to $T_i$ and $T_j$. Let $D^d_{i,j}(u)=\{ v\in V(D_{i,j}) |\ d_{D_{i,j}}(u,v)=d\}$. 
We add to $T_i$ the edges of the set $\{ vv'\in E(D_{i,j})|\ v\in D^d_{i,j}(u),\ v'\in D^{d+1}_{i,j}(u),\ d\text{ is even} \}$ and to $T_j$ the edges of the set  $\{ vv'\in E(D_{i,j})|\ v\in D^d_{i,j}(u),\ v'\in D^{d+1}_{i,j}(u),\ d\text{ is odd} \}$. 
Second, if $D_{i,j}$ is not a tree, then we suppose that $u$ is in a cycle of $D_{i,j}$. Let $e$ be an edge of this cycle incident with $u$ and let $T_{i,j}$ be a spanning tree of $D_{i,j}-e$. We define $B^d_{i,j}(u)$ as follows: $\{ v\in V(D_{i,j})  |\ d_{T_{i,j}}(u,v)=d\}$. 
We add to $T_i$ the edges of the set  $\{ vv'\in E(T_{i,j})|\ v\in B^d_{i,j}(u),\ v'\in B^{d+1}_{i,j}(u),\ d\text{ is even} \}$ and to $T_j$ the edges of the set  $\{ vv'\in E(T_{i,j})|\ v\in B^d_{i,j}(u),\ v'\in B^{d+1}_{i,j}(u),\ d\text{ is odd} \}\cup \{e\}$. We repeat this process for every connected component of $B(V_i, V_j)$ and every two integers $i$ and $j$, $1\le i< j\le k$. Since there is only one common edge between $T_i$ and $T_j$ for each connected component that is a tree and since $\sum_{1\le i<j\le k} c_{i,j}\le \ell$, the set $E(T_1,\ldots, T_k)$ contains at most $\ell$ edges. Therefore, we obtain, by Property ii), $k$ $(0,\ell)$-disjoint spanning trees.

Let us prove the converse of the previous implication. Suppose there exist $k$ $(0,\ell)$-disjoint spanning trees $T_1, \ldots, T_k$ in $G$. The set $I(T_i)$, $1\le i\le k$, induces a connected subgraph in $G$. We begin by setting $V_i=I(T_i)$, for each integer $i$, $1\le i\le k$. If some vertices are inner vertices in no trees, we can add them to any set among $V_1, \ldots, V_k$. Thus, Property i) follows. Let $i$ and $j$ be two integers, $1\le i<j \le k$. Suppose there exists one isolated vertex $u$ in $B(V_i,V_j)$. Without loss of generality, suppose $u\in V_i$. By Proposition \ref{prop1}.i), we obtain a contradiction since $u\notin I(T_j)$ and $u$ has no neighbor in $I(T_j)$. Thus, Property ii) follows.
Now suppose $\sum_{1\le i<j\le k} c_{i,j}>\ell$. Let $D_{i,j}$ be a connected component which is a tree in $B(V_i,V_j)$ for some integers $i$ and $j$ and suppose that $D_{i,j}$ contains no edge from $E(T_1,\ldots, T_k)$. Since $D_{i,j}$ has $|V(D_{i,j})|-1$ edges, it is impossible that every vertex of $V(D_{i,j})\cap V_i$ is adjacent to a vertex of $V(D_{i,j})\cap V_j$ in $T_j$ and that every vertex of $V(D_{i,j})\cap V_j$ is adjacent to a vertex of $V(D_{i,j})\cap V_i$ in $T_i$, since it would require $|V(D_{i,j})|$ edges. Thus, for every two integers $i$ and $j$ and every connected component $D_{i,j}$ of $B(V_i, V_j)$, if $D_{i,j}$ is a tree then $V(D_{i,j})\cap E(T_1,\ldots, T_k)\neq \emptyset$ and we obtain a contradiction since $\sum_{1\le i<j\le k} c_{i,j}>\ell$ implies $|E(T_1,\ldots, T_k)|>\ell$. Consequently, Property iii) follows.
\end{proof}
\subsection{$(\ell,*)$-disjoint spanning trees}

For a graph $G$ and a subset of vertices $A\subseteq V(G)$, let $N(A)=\{u\in V(G)\setminus A|\ uv\in E(G), \ v\in A\}$.
In a similar way than Zelinka \cite{ZE1986}, we prove that the notion of $\ell$-rooted connected dominating sets is equivalent to a notion of partition.
\begin{de}
An \emph{$\ell$-rooted partition} of $G$ into $k+1$ sets is a partition of $V(G)$ into $k+1$ sets of vertices $V_1,\ldots,V_k,A$ such that:
\begin{enumerate}
\item[i)] $|A|\le \ell$;
\item[ii)] $G[V_i\cup A]$ is connected, for each integer $i$, $1\le i\le k$;
\item[iii)] $B(V_i,V_j)- N(A)$ contains no isolated vertex, for every two integers $i$ and $j$, $1\le i<j\le k$.
\end{enumerate}
\end{de}
Figure \ref{despart} illustrates an $1$-rooted partition on a specific graph.
\begin{theo}
Let $G$ be a graph. There exist $k$ $\ell$-rooted connected dominating sets $D_1, \ldots, D_k$ in $G$ if and only if $G$ has an $\ell$-rooted partition into $k+1$ sets.
\end{theo}
\begin{proof}
Suppose $G$ has an $\ell$-rooted partition into $k+1$ sets $V_1, \ldots,V_k, A$. We begin by setting $D_i= V_i\cup A$ for each integer $i$, $1\le i\le k$. Since $V_1, \ldots ,V_k,A$ is a partition, we have $|\cup_{1\le i< j\le k} D_i \cap D_j|\le \ell$. Moreover, by Property ii), the subgraphs induced by the sets $D_1$,\ldots, $D_k$ are all connected. It remains to prove that $D_i$ is a dominating set, for each integer $i$, $1\le i \le k$. Since the vertices of $N(A)$ are already dominated by a vertex of $A\subseteq D_i$, for each integer $i$, Property iii) implies that every vertex of $V(G)\setminus (V_i\cup N(A))$ has a neighbor in $V_i$, for each integer $i$.

Suppose there exist $k$ $\ell$-rooted connected dominating sets $D_1,\ldots, D_k$ in $G$. We begin by setting $A=\cup_{1\le i< j\le k} D_i \cap D_j$. Afterward, we set $V_i=D_i \setminus A$, for each integer $i$, $1\le i\le k$. By definition, Property i) and Property ii) are satisfied by $V_1, \ldots ,V_k,A$. It remains to prove Property iii). By contradiction, suppose that a vertex $u\in V_i$ has no neighbor in $V_j\cup A$, for some integers $i$ and $j$. This fact implies that $D_j$ is not a dominating set and Property iii) follows.
\end{proof}
In the following definition we introduce the construction of a graph denoted by $G(k,A)$.
\begin{de}
Let $G$ be a graph, $k$ be an integer and $A=\{u_1,\ldots,u_\ell\}\subseteq V(G)$ be a subset of vertices. We denote by $G(k,A)$ the  graph obtained by replacing one by one each vertex $u_i$, for $1\le i\le \ell$, by a complete graph of order $k$, and by adding edges between each vertex of this clique and every vertex of $N(u_i)$.
\end{de}
We finish by proving that determining if a graph $G$ contains $k$ $\ell$-rooted connected dominating sets is equivalent to determine if the graph $G(k,A)$ has $\ell$ disjoint connected dominating sets, for some subset of vertices $A\subseteq V(G)$. In contrast with the two previous propositions, this alternative characterization is expressed in terms of disjoint dominating sets.
\begin{prop}
There exist k $\ell$-rooted connected dominating sets $D_1,\ldots, D_k$ in a graph $G$ if and only if there exist a subset of vertices $A\subseteq V(G)$ such that $|A|\le \ell$ and $d_c(G(k,A))\ge k$.
\end{prop}
\begin{proof}
Let $G$ be a graph.
Suppose there exist $k$ $\ell$-rooted connected dominating sets $D_1,\ldots, D_k$ in $G$. We begin by setting $A=\cup_{1\le i< j\le k} D_i \cap D_j$. 
Let $K^i_k$ denote the clique from $G(k,A)$ which replaces the vertex $u_i$ in $G(k,A)$, for $1\le i\le \ell$.
We can construct $k$ disjoint connected dominating sets $D'_1, \ldots, D'_k$ in $G(k,A)$ as follows: for each integer $j$, $1\le j\le k$, $D'_j$ contains the vertices from $D_j\setminus A$ and one different vertex by clique $K^i_k$, for each integer $i$, $1\le i\le \ell$.

Suppose there exist a subset of vertices $A\subseteq V(G)$ such that $|A|\le \ell$ and $d_c(G(k,A))\ge k$. Let $D'_1, \ldots, D'_k$ be disjoint connected dominating sets in $G(k,A)$. We can construct $k$  $\ell$-rooted connected dominating sets $D_1, \ldots, D_k$ in $G$ as follows: for each integer $j$, $1\le j\le k$, $D_j$ contains the vertices from $D'_j \setminus (K^1_k\cup\ldots \cup K_k^\ell)\cup \{u_1, \ldots, u_\ell\}$.

\end{proof}
\section{An NP-complete problem for every integers $i$ and $j$}

We define the following decision problem:
\begin{center}
\parbox{10cm}{
\setlength{\parskip}{.05cm}
\textbf{$k$-$(i,j)$-DSP}

\textbf{Instance} : A graph $G$.

\textbf{Question}: Does there exist $k$ $(i,j)$-disjoint spanning trees in $G$ ?}

\end{center}
\begin{theo}\label{comp1}
Let $i$ and $j$ be non negative integers.
The problem $2$-$(i,j)$-DSP is an NP-complete problem for every pair of integer $(i,j)$. 
\end{theo}
\begin{proof}
Hasunuma \cite{HA2002} has proved that the following problem is NP-complete:
\begin{center}
\parbox{10cm}{
\setlength{\parskip}{.05cm}
\textbf{$2$-$(u,v)$-CIST}

\textbf{Instance} : A graph $G$ and two vertices $u$ and $v$ of $V(G)$.

\textbf{Question}: Does there exist two completely independent spanning trees $T_1$ and $T_2$ in $G$ with $u\in I(T_1)$ and $v\in I(T_2)$ ?}

\end{center}
Initially, the NP-complete problem considered by Hasunuma \cite{HA2002} consists in determining if there exist two completely independent spanning trees in a graph $G$.
However, by analyzing Hasunuma's reduction we can also obtain that the problem $2$-$(u,v)$-CIST is NP-complete by using the same reduction (it suffices to consider that $u$ is $v_B$ and that $v$ is $v_R$ in Hasunuma's reduction). Also, it is trivial to prove that the problem $2$-$(i,j)$-DST is in NP since the description of two spanning trees in a graph $G$ (when $dst_{i,j}(G)\ge2$) ensures the existence of these two $(i,j)$-disjoint spanning trees.
We use a reduction from $2$-$(u,v)$-CIST.

We introduce the three following operations that will be useful to describe our reduction:
\begin{enumerate}
\item[i)] $H$-add is an operation on a graph with two prescribed vertices $w_1$ and $w_2$ that consists in adding the graph $H$ from Figure \ref{reduc} and identifying $w_1$ with $p_1$ and $w_2$ with $p_2$;
\item[ii)] $H'$-add is an operation on a graph with two prescribed vertices $w_1$ and $w_2$ that consists in adding the graph $H'$ from Figure \ref{reduc} and identifying $w_1$ with $p_1$ and $w_2$ with $p_2$;
\item[ii)] $H^{+}$-add is an operation on a graph with two prescribed vertices $w_1$ and $w_2$ that consists in adding the graph $H^{+}$ from Figure \ref{reduc} and identifying $w_1$ with $p_1$ and $w_2$ with $p_2$;
\end{enumerate}
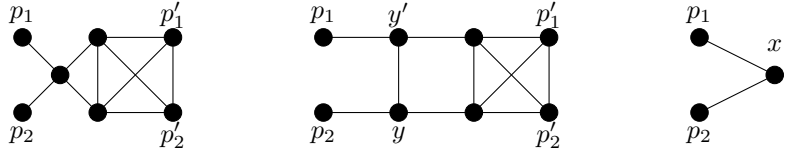
\begin{figure}[t]
\begin{center}
\begin{tikzpicture}
\draw (0,0) -- (0.5,0.5);
\draw (0,1) -- (0.5,0.5);
\draw (1,0) -- (0.5,0.5);
\draw (1,1) -- (0.5,0.5);
\draw (1,0) -- (2,0);
\draw (1,0) -- (1,1);
\draw (2,0) -- (2,1);
\draw (1,1) -- (2,1);
\draw (1,0) -- (2,1);
\draw (1,1) -- (2,0);

\draw (4,0) -- (7,0);
\draw (4,1) -- (7,1);
\draw (5,0) -- (5,1);
\draw (6,0) -- (7,1);
\draw (6,1) -- (7,0);
\draw (6,0) -- (6,1);
\draw (7,0) -- (7,1);

\draw (9,0) -- (10,0.5);
\draw (9,1) -- (10,0.5);

\node at (0,0) [circle,draw=black,fill=black,scale=0.7]{};
\node at (0,1) [circle,draw=black,fill=black,scale=0.7]{};
\node at (0.5,0.5) [circle,draw=black,fill=black,scale=0.7]{};
\node at (1,0) [circle,draw=black,fill=black,scale=0.7]{};
\node at (1,1) [circle,draw=black,fill=black,scale=0.7]{};
\node at (2,0) [circle,draw=black,fill=black,scale=0.7]{};
\node at (2,1) [circle,draw=black,fill=black,scale=0.7]{};
\node at (0,-0.3) {$p_2$};
\node at (0,1.3) {$p_1$};
\node at (2,-0.3) {$p'_2$};
\node at (2,1.3) {$p'_1$};

\node at (4,0) [circle,draw=black,fill=black,scale=0.7]{};
\node at (4,1) [circle,draw=black,fill=black,scale=0.7]{};
\node at (5,0) [circle,draw=black,fill=black,scale=0.7]{};
\node at (5,1) [circle,draw=black,fill=black,scale=0.7]{};
\node at (6,0) [circle,draw=black,fill=black,scale=0.7]{};
\node at (6,1) [circle,draw=black,fill=black,scale=0.7]{};
\node at (7,0) [circle,draw=black,fill=black,scale=0.7]{};
\node at (7,1) [circle,draw=black,fill=black,scale=0.7]{};
\node at (4,-0.3) {$p_2$};
\node at (4,1.3) {$p_1$};
\node at (5,-0.3) {$y$};
\node at (5,1.3) {$y'$};
\node at (7,-0.3) {$p'_2$};
\node at (7,1.3) {$p'_1$};

\node at (9,0) [circle,draw=black,fill=black,scale=0.7]{};
\node at (9,1) [circle,draw=black,fill=black,scale=0.7]{};
\node at (10,0.5) [circle,draw=black,fill=black,scale=0.7]{};
\node at (9,-0.3) {$p_2$};
\node at (9,1.3) {$p_1$};
\node at (10,0.9) {$x$};

\end{tikzpicture}
\end{center}
\caption{The graph $H$ (on the left), the graph $H'$ (on the middle) and the graph $H^{+}$ (on the right).}
\label{reduc}
\end{figure}

\begin{figure}[t]
\begin{center}
\begin{tikzpicture}
\draw[ultra thick, style=dashed,color=red]  (0,0) -- (0.5,0.5);
\draw[ultra thick,color=blue] (0,1) -- (0.5,0.5);
\draw[ultra thick, style=dashed,color=red]  (1,0) -- (0.5,0.5);
\draw[ultra thick,color=blue] (1,1) -- (0.5,0.5);
\draw[ultra thick, style=dashed,color=red] (1,0) -- (2,0);
\draw[ultra thick, style=dashed,color=red]  (1,0) -- (1,1);
\draw[ultra thick, style=dashed,color=red]  (2,0) -- (2,1);
\draw[ultra thick,color=blue] (1,1) -- (2,1);
\draw[ultra thick,color=blue] (1,0) -- (2,1);
\draw[ultra thick,color=blue] (1,1) -- (2,0);

\draw[ultra thick, style=dashed,color=red] (4,0) -- (7,0);
\draw[ultra thick,color=blue] (4,1) -- (7,1);
\draw[ultra thick,color=blue] (5,0) .. controls (4.9,0.5) .. (5,1);
\draw[ultra thick, style=dashed,color=red] (5,0) .. controls (5.1,0.5) .. (5,1);
\draw[ultra thick,color=blue] (6,0) -- (7,1);
\draw[ultra thick,color=blue] (6,1) -- (7,0);
\draw[ultra thick, style=dashed,color=red](6,0) -- (6,1);
\draw[ultra thick, style=dashed,color=red] (7,0) -- (7,1);

\draw[ultra thick, style=dashed,color=red]  (9,0) -- (10,0.5);
\draw[ultra thick,color=blue] (9,1) -- (10,0.5);

\node at (0,0) [regular polygon, regular polygon sides=3,draw=black,fill=red,scale=0.5]{};
\node at (0,1) [circle,draw=black,fill=blue,scale=0.7]{};
\node at (0.5,0.5) [regular polygon, regular polygon sides=4,draw=black,fill=purple,scale=0.7]{};
\node at (1,0) [regular polygon, regular polygon sides=3,draw=black,fill=red,scale=0.5]{};
\node at (1,1) [circle,draw=black,fill=blue,scale=0.7]{};
\node at (2,0) [regular polygon, regular polygon sides=3,draw=black,fill=red,scale=0.5]{};
\node at (2,1) [circle,draw=black,fill=blue,scale=0.7]{};

\node at (4,0) [regular polygon, regular polygon sides=3,draw=black,fill=red,scale=0.5]{};
\node at (4,1) [circle,draw=black,fill=blue,scale=0.7]{};
\node at (5,0) [regular polygon, regular polygon sides=3,draw=black,fill=red,scale=0.5]{};
\node at (5,1) [circle,draw=black,fill=blue,scale=0.7]{};
\node at (6,0) [regular polygon, regular polygon sides=3,draw=black,fill=red,scale=0.5]{};;
\node at (6,1) [circle,draw=black,fill=blue,scale=0.7]{};
\node at (7,0) [regular polygon, regular polygon sides=3,draw=black,fill=red,scale=0.5]{};;
\node at (7,1) [circle,draw=black,fill=blue,scale=0.7]{};

\node at (9,0) [regular polygon, regular polygon sides=3,draw=black,fill=red,scale=0.5]{};
\node at (9,1) [circle,draw=black,fill=blue,scale=0.7]{};
\node at (10,0.5) [circle,draw=black,fill=black,scale=0.7]{};

\end{tikzpicture}
\end{center}
\caption{Pattern to construct the trees in $H$ (on the left), the graph $H'$ (on the middle) and the graph $H^{+}$ (on the right) (simple line: edge of $T_1$; dashed line: edge of $T_2$; boxed vertices: inner vertices of both $T_1$ and $T_2$).}
\label{reduc2}
\end{figure}
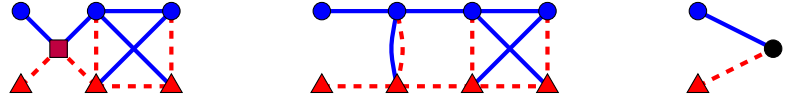
Let $G$ be a graph and let $u$ and $v$ be two vertices of $V(G)$.
We construct a graph $G'$ from $G$ as follows.
Let $\ell\ge 1$ be a positive integer.
We begin by constructing two graphs $H_{\ell}$ and $H'_{\ell}$ by induction. The graph $H_1$ is the graph $H$ and the graph $H'_1$ is the graph $H'$. The graph $H_{\ell+1}$ is obtained from $H_{\ell}$ by doing an $H$-add on the two vertices of degree $3$ in $H_\ell$ (denoted by $p'_1$ and $p'_2$ in the left part of Figure \ref{reduc}, for $\ell=1$). The graph $H'_{\ell+1}$ is obtained from $H'_\ell$ by doing an $H'$-add on the two vertices of degree $3$ (also denoted by $p'_1$ and $p'_2$ in the middle part of Figure \ref{reduc}, for $\ell=1$). In $H_\ell$ and $H'_\ell$, we denote by $p'_1$ and $p'_2$ the two remaining vertices of degree $3$.

Finally, the graph $H_{i,j}$ is obtained by taking $H_i$ and $H'_j$, identifying $p'_1$ in $H_i$ with a vertex of degree one in $H'_j$ and $p'_2$ in $H_i$ with the other vertex of degree one in $H'_j$ and doing a $H^{+}$-add on the two vertices of degree $3$ which are labeled by $p'_1$ and $p'_2$ in $H'_j$.
The graph $G'$ is obtained by taking a copy of $G$, adding $H_{i,j}$, identifying the vertex $u$ with a vertex of degree one in $H_{i,j}$ and identifying the vertex $v$ with the other vertex of degree one in $H_{i,j}$.

Suppose there exist two completely independent spanning trees $T_1$ and $T_2$ in $G$ with $u\in I(T_1)$ and $v\in I(T_2)$.
We can construct two $(i,j)$-disjoint spanning trees in $G'$ by reproducing the trees $T_1$ and $T_2$ in the graph $G'$ restricted to $G$ and by using the patterns described in Figure \ref{reduc2} in order to extend the spanning trees to $H_{i,j}$.

Suppose there exist two $(i,j)$-disjoint spanning trees  $T_1$ and $T_2$ in $G'$. 
Note that there are $i$ articulation vertices in the graph $G'$ restricted to $H_{i,j}$. These $i$ articulation vertices should be inner vertices in both $T_1$ and $T_2$. Thus, the trees $T_1$ and $T_2$ restricted to $G$ are internally vertex-disjoint.
By Proposition \ref{prop1}, the vertex $x$ (obtained by $H^{+}$-add in $H_{i,j}$ and illustrated in Figure \ref{reduc}) should be adjacent to an inner vertex of $T_1$ and to an inner vertex of $T_2$. Thus in order that $T_1$ and $T_2$ be connected, there must be a path from $x$ to $u$ in $T_1$ and a path from $x$ to $v$ in $T_2$ (we can exchange $T_1$ and $T_2$ if necessary). Note that the vertices $y$ and $y'$ from a copy of $H'$ (illustrated in Figure \ref{reduc}) cannot be both inner vertices of the same tree since it would be impossible to have a path from $x$ to $u$ in $T_1$ and another path from $x$ to $v$ in $T_2$. Thus, in order that $y$ and $y'$ belong to both $T_1$ and $T_2$ , the edge $yy'$ should belong to both $T_1$ and $T_2$ (since we already have $i$ articulation vertices). Moreover, since there are $j$ copies of $H'$ in the graph $H_{i,j}$, the trees $T_1$ and $T_2$ restricted to $G$ are both internally vertex-disjoint and edge-disjoint.

\end{proof}
\section{Sufficient conditions to have $(i,j)$-disjoint spanning trees}
\subsection{$k$-connectivity}
We begin this section by proving classical properties about cut sets.
\begin{prop}
Let $G$ be a graph and let $T_1, \ldots, T_k$ be $(i,j)$-disjoint spanning trees in $G$.
For every subset of vertices $A\subseteq V(G)$ such that $|A|< k$ and $G-A$ is not connected, (at least) one vertex of $A$ is in $I(T_1,\ldots T_k)$.
For every subset of edges $B\subseteq E(G)$ such that such that $|B|< k$ and $G-B$ is not connected, (at least) one edge of $B$ is in $E(T_1,\ldots T_k)$.
\end{prop}
\begin{proof}
Let $A\subseteq V(G)$ be a subset of vertices such that $|A|< k$ and $G-A$ is not connected.
Remark that $I(T_{\ell})\cap A$ should not be empty, for every integer $\ell$, $1\le \ell \le k$, since it would imply that $T_\ell$ is not connected. Since $|A|< k$, a vertex of $A$ should be in $I(T_1,\ldots T_k)$. The same property holds for $B$.
\end{proof}
\begin{prop}
Let $G$ be a graph and let $T_1, \ldots, T_k$ be $(i,j)$-disjoint spanning trees in $G$. 
Let $a$ be the number of articulation vertices which do not belong to bridges in $G$ and let $b$ be the number of bridges in $G$. 
We have $i\ge a+2b$ and $j\ge b$.
\end{prop}
\begin{proof}
Since an articulation vertex belongs to every spanning tree of $G$, we have $i\ge a$.
The same goes for the bridges and their extremities.
\end{proof}
Since the presence of a $k$-cut in a graph $G$ implies that there do not exist $k+1$ disjoint connected dominating set, it is natural to ask whether a $k$-connected graph, for $k$ sufficiently large, contains at least two disjoint connected dominating sets \cite{HE1984}.
In the paper in which completely independent spanning trees have been introduced \cite{HA2001}, the same question has been asked for two completely independent spanning trees.

Using the construction from Kriesell \cite{KR1999} or Péterfalvi \cite{FE2011}, we can obtain a family of $k$-connected graphs that do not contain two completely independent spanning trees. We recall the construction of the family of graphs considered by Kriesell \cite{KR1999}.

\begin{de}[\cite{KR1999}]
Let $k$ and $\ell$ be two integers such that $\ell \ge k$.
Let $G_{k,\ell}$ be the bipartite graph with vertex set $\{1,\ldots, \ell\}\cup \{u_{A}|\ A\subseteq \{1,\ldots \ell\},\ |A|=k\}$ and edge set $\{ i u_A |\ i\in A,\ u_A\in V(G_{k,\ell}) ,\ 1\le i \le \ell \}$.
The graph $G_{k,\ell}$ corresponds to the incidence graph of the complete $k$-uniform hypergraph with $\ell$ vertices.
\end{de}
Note that the graph $G_{k,\ell}$ is $k$-connected and bipartite.
Using a similar proof than that of Kriesell, we obtain the following theorem which shows that there exist $k$-connected graphs which do not contain two $(i,j)$-disjoint spanning trees, for every three positive integers $k\ge 2$, $i$ and $j$.
\begin{theo}
Let $i$, $j$ and $k\ge 2$ be integers.
For any $\ell \ge 2k+i-1$, the graph $G_{k,\ell}$ does not contain two $(i,j)$-disjoint spanning trees.
\end{theo}
\begin{proof}
Suppose there exist two $(i,j)$-disjoint spanning trees $T_1$ and $T_2$ in $G_{k,\ell}$.
Let $H_1$ and $H_2$ be the two subsets of vertices forming a bipartition of $G_{k,\ell}$, with $H_1=\{1,\ldots, \ell\}$ and $ H_2=\{u_{A}|\ A\subseteq \{1,\ldots \ell\},\ |A|=k\}$.
Let $B=I(T_1,T_2)\cap H_1$.
Note that, by definition of $(i,j)$-disjoint spanning trees, $|B|\le i$. We consider a set $A_0\subset H_1\setminus B$, $|A_{0}|=k$.
By Proposition \ref{prop1}, at least one inner vertex of $T_1$ is adjacent to $u_{A_0}$. This inner vertex of $T_1$ is denoted by $v_{0}$. Inductively, since $|H_1|=\ell\ge 2k+i-1$, for $1\le q\le k-1$ , we can create a set $A_q\subseteq H_1 \setminus (B\cup \{u_0,\ldots,u_{q-1}\})$ with $|A_{q}|=k$ and obtain that there exists a vertex $v_q\in A_q\cap I(T_1)$ adjacent to $u_{A_q}$. The set $D=\{v_0,\ldots, v_{k-1} \}\subseteq H_1$ is such that $|D|=k$ and $D\subseteq I(T_1)$. Hence, we have a contradiction with Proposition \ref{prop1}.ii), since $u_{D}$ has no neighbor which is a inner vertex of $T_2$ and $G_{k,\ell}$ has diameter greater than $2$ when $\ell>k$.
\end{proof}
\subsection{Dirac's and Ore's conditions}
We begin this subsection by proving that there are at least two disjoint dominating sets in some particular graphs.
\begin{prop}\label{yololo}
There exist two disjoint connected dominating sets in $C_4$ and three disjoint connected dominating sets in $K_{3,3}$
\end{prop}
\begin{proof}
These disjoint connected dominating sets are illustrated in Figure \ref{C4K33}.
\end{proof}
\begin{figure}[t]
\begin{center}
\begin{tikzpicture}
\draw (0,0) -- (1,0);
\draw (0,0) -- (0,1);
\draw (1,0) -- (1,1);
\draw (0,1) -- (1,1);
\node at (0,0) [circle,draw=blue,fill=blue,scale=0.7]{};
\node at (0,1) [circle,draw=blue,fill=blue,scale=0.7]{};
\node at (1,0) [regular polygon, regular polygon sides=3,draw=black,fill=red,scale=0.5]{};
\node at (1,1) [regular polygon, regular polygon sides=3,draw=black,fill=red,scale=0.5]{};

\draw (4,0) -- (4,1);
\draw (4,0) -- (5,1);
\draw (4,0) -- (6,1);
\draw (5,0) -- (4,1);
\draw (5,0) -- (5,1);
\draw (5,0) -- (6,1);
\draw (6,0) -- (4,1);
\draw (6,0) -- (5,1);
\draw (6,0) -- (6,1);
\node at (4,0) [circle,draw=blue,fill=blue,scale=0.7]{};
\node at (4,1) [circle,draw=blue,fill=blue,scale=0.7]{};
\node at (5,0) [regular polygon, regular polygon sides=3,draw=black,fill=red,scale=0.5]{};
\node at (5,1) [regular polygon, regular polygon sides=3,draw=black,fill=red,scale=0.5]{};
\node at (6,0) [regular polygon, regular polygon sides=4,draw=black,fill=green,scale=0.7]{};
\node at (6,1)  [regular polygon, regular polygon sides=4,draw=black,fill=green,scale=0.7]{};
\end{tikzpicture}
\end{center}
\caption{Disjoint connected dominating sets in $C_4$ (on the left) and in $K_{3,3}$ (on the right) (circles: $D_1$; triangles: $D_2$; squares: $D_3$).}
\label{C4K33}
\end{figure}
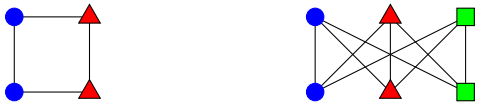
A graph $G$ satisfies the condition of Dirac if $\delta(G)\ge |V(G)|/2$ and satisfies the condition of Ore if $\min \{ d(u)+d(v) |\ uv\notin E(G)\} \ge |V(G)|$. Araki \cite{AR2014} proved that every graph $G$ with $|V(G)|\ge 7$ satisfying Dirac's condition contains two completely independent spanning trees. Moreover, Fan, Hong and Liu \cite{FA2014} proved that every graph $G$ with $|V(G)|\ge 7$ satisfying Ore's condition contains two completely independent spanning trees. The only graphs with $|V(G)|< 7$ satisfying the Dirac condition or the Ore condition which do not contain two completely independent spanning trees are $P_2$, $C_4$ and $K_{3,3}$.
Thus, by Proposition \ref{yololo}, we obtain the two following theorems:
\begin{theo}[\cite{AR2014}]\label{cond1}
Let $G$ be a graph. If $\delta(G)\ge |V(G)|/2$, then there exist two disjoint connected dominating sets.
\end{theo}
\begin{theo}[\cite{FA2014}]\label{cond2}
Let $G$ be a graph. If $\min \{ d(u)+d(v)|\ uv\notin E(G)\} \ge |V(G)|$, then there exist two disjoint connected dominating sets.
\end{theo}
Moreover, there exists a graph of order $n$ satisfying $\delta(G)\ge \lceil n/2 \rceil -1$ and $\min \{ d(u)+d(v)|\ uv\notin E(G)\} \ge n-1$, that does not contain two disjoint connected dominating sets. Such graph can be constructed by taking two complete graphs $K_{\lfloor (n+1)/2 \rfloor}$ and $K_{\lceil (n+1)/2 \rceil}$, for $n$ a positive integer, and by identifying a vertex of the first clique with a vertex of the second clique. This fact implies that the bounds in the previous theorems are tight. It could be possible to improve the recent results about Dirac's condition \cite{CH2015,HA2016,HO2016} by only considering disjoint connected dominating sets.

\section{Number of inner vertices and edges in $(i,j)$-disjoint spanning trees}
\subsection{Required number of edges}
We begin this section by giving necessary conditions on the number of edges of a graph $G$ in order to have $k$ $(i,j)$-disjoint spanning trees.
\begin{prop}\label{prop44}
Let $G$ be a graph of order $n$ and let $T_1, \ldots, T_k$ be $(i,j)$-disjoint spanning trees in $G$. We have $|E(G)|\ge k(n-1)-j(k-1)$.
\end{prop}
\begin{proof}
Suppose $G$ contains at least $k$ $(i,j)$-disjoint spanning trees. Since every spanning tree contains $n-1$ edges and 
since an edge in $E(T_1, \ldots T_k)$ can be in at most $k$ trees, we obtain that $G$ contains at least $k(n-1)-j(k-1)$ edges.
\end{proof}
Note that the grid $G(2,n)$ satisfies the equality for $k=2$ and $j=n$.
This last proposition can be improved for $i=0$ \cite{HAR2001} since, by Proposition \ref{megaprop}.i), an edge in $E(T_1, \ldots T_k)$ can be in at most two trees.

\begin{cor}\label{prop444}
Let $G$ be a graph of order $n$ and let $T_1, \ldots, T_k$ be $(0,j)$-disjoint spanning trees in $G$. We have $|E(G)|\ge k(n-1)-j$.
\end{cor}
Moreover, for an arbitrary large $j$, the following bound is known.
\begin{prop}\cite{HAR2001}
A graph $G$ of order $n$ such that $d_c(G)\ge k$ has at least $n(k+1)/2-k$ edges. This bound is sharp since $d_c(K_{k,k})=k$.
\end{prop}

\subsection{Distribution of the inner vertices}
The following observation illustrates the existence of an $(i,j)$-disjoint spanning tree with possibly less inner vertices than the others.
\begin{obs}
Let $G$ be a graph of order $n$ and let $T_1, \ldots, T_k$ be $(i,j)$-disjoint spanning trees in $G$. There exists a tree $T$ among $T_1, \ldots, T_k$ satisfying $|I(T)|\le \lfloor (n-i)/k\rfloor+i$.
\end{obs}
Two sets of vertices $V_1$ and $V_2$ are \emph{balanced} if $||V_1|-|V_2||\le 1$.
We begin by proving that there exists a graph $G$ satisfying $d_c(G)\ge2$ but in which no two disjoint connected dominating sets are balanced.
Let $P_n^{*}$ be the graph constructed by taking one copy of $P_n$, by adding a new vertex $u$ and by adding the edges between $u$ and the vertices of $P_n$. Figure \ref{PnPn} illustrates the graph $P_n^{*}$ for $n=6$.
\begin{figure}[t]
\begin{center}
\begin{tikzpicture}
\draw (0,0) -- (5,0);
\draw (0,0) -- (2.5,1);
\draw (1,0) -- (2.5,1);
\draw (2,0) -- (2.5,1);
\draw (3,0) -- (2.5,1);
\draw (4,0) -- (2.5,1);
\draw (5,0) -- (2.5,1);
\node at (0,0) [circle,draw=black,fill=black,scale=0.7]{};
\node at (1,0) [circle,draw=black,fill=black,scale=0.7]{};
\node at (2,0) [circle,draw=black,fill=black,scale=0.7]{};
\node at (3,0) [circle,draw=black,fill=black,scale=0.7]{};
\node at (4,0) [circle,draw=black,fill=black,scale=0.7]{};
\node at (5,0) [circle,draw=black,fill=black,scale=0.7]{};
\node at (2.5,1) [circle,draw=black,fill=black,scale=0.7]{};
\node at (2.5,1.3) {$u$};
\draw (0+7,0) -- (5+7,0);
\draw (0+7,0) -- (2+7,1);
\draw (1+7,0) -- (2+7,1);
\draw (2+7,0) -- (2+7,1);
\draw (3+7,0) -- (2+7,1);
\draw (4+7,0) -- (2+7,1);
\draw (5+7,0) -- (2+7,1);
\draw (2+7,1) -- (3+7,1);
\draw (0+7,0) -- (3+7,1);
\draw (5+7,0) -- (3+7,1);
\node at (0+7,0) [circle,draw=black,fill=black,scale=0.7]{};
\node at (1+7,0) [circle,draw=black,fill=black,scale=0.7]{};
\node at (2+7,0) [circle,draw=black,fill=black,scale=0.7]{};
\node at (3+7,0) [circle,draw=black,fill=black,scale=0.7]{};
\node at (4+7,0) [circle,draw=black,fill=black,scale=0.7]{};
\node at (5+7,0) [circle,draw=black,fill=black,scale=0.7]{};
\node at (2+7,1) [circle,draw=black,fill=black,scale=0.7]{};
\node at (3+7,1) [circle,draw=black,fill=black,scale=0.7]{};
\node at (2+7,1.3) {$u$};
\node at (3+7,1.3) {$v$};
\end{tikzpicture}
\end{center}
\caption{The graph $P_6^{*}$ (on the left) and the graph $P_6^{-}$ (on the right).}
\label{PnPn}
\end{figure}
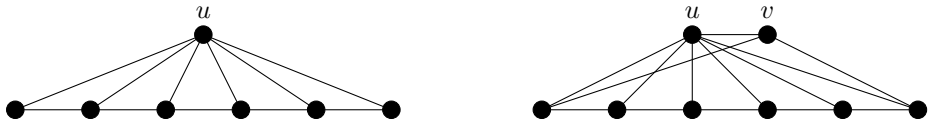

\begin{prop}
Let $n\ge 5$.
For any two disjoint connected dominating sets $D_1$ and $D_2$ in $P_n^{*}$, $||D_1|-|D_2||\ge n-5$.
\end{prop}
\begin{proof}
Suppose without loss of generality that $u\notin D_1$. Since $D_1$ should be connected, it should contain consecutive vertices of $P_n$. Moreover, since $D_1$ should be dominating, it should contain every vertex of $P_n$, except its extremities. Thus, $|D_1|\ge n-2$ and consequently $|D_2|\le 3$. Therefore, we have $||D_1|-|D_2||\ge n-5$.
\end{proof}

Note that the graph $P_n^{*}$ does not contain two completely independent spanning trees. Thus, it could be true that every graph containing two completely independent spanning trees contains two completely independent spanning trees $T_1$ and $T_2$ such that $||I(T_1)|-|I(T_2)||\le 1$. However, the following proposition illustrates that it is not the case. Let $P_n^{+}$ be the graph obtained by taking one copy of $P^{*}_n$, by adding a new vertex $v$ and by adding the edge $uv$ and the edges between $v$ and the extremities of $P_n$, $u$ being the vertex of maximal degree in $P_n^{*}$, $P_n$ being the induced path of $n$ vertices in $P_n^{*}$ obtained by removing $u$. Figure \ref{PnPn} illustrates the graph $P_n^{+}$ for $n=6$.
\begin{prop}
Let $n\ge 3$.
For any two completely independent spanning trees $T_1$ and $T_2$ in $P_n^{+}$, $||I(T_1)|-|I(T_2)||\ge n-2$.
\end{prop}
\begin{proof}
First, observe that there exist two completely independent spanning trees in $P_n^{+}$ since $\{u,v\}$ and $V(P_n^{+})\setminus \{u,v\}$ is a $0$-CIST-partition.
Now, suppose there exist two completely independent spanning trees $T_1$ and $T_2$ and suppose without loss of generality that $u\notin I(T_1)$. Since the graph induced by the vertices of $I(T_1)$ should be connected, it should contain consecutive vertices of $P_n$. Moreover, $I(T_1)$ should be dominating set. Since either $\{u\}$ and the subsets of $V(P_n^{+})\setminus \{u\}$ or $\{u,v\}$ and the proper subsets of $V(P_n^{+})\setminus \{u,v\}$ do not form a $0$-CIST partition, we have $I(T_1)=V(P_n^{+})\setminus \{u,v\}$ and $I(T_2)=\{u,v\}$. Therefore, we have $||I(T_1)|-|I(T_2)||\ge n-2$.
\end{proof}
Even if there exist graphs only containing two non-balanced disjoint connected dominating sets, it could be interesting to find classes of graphs for which there always exist two disjoint connected dominating sets which are balanced. For example, the class of graphs with minimum degree at least $|V(G)|/2$, is such a class \cite{HO2016}.
\section{$(i,j)$-disjoint spanning trees in some simple classes of graphs}
\subsection{Square of graphs}
The square of a graph $G$, denoted by $G^2$, is the graph obtained from $G$ by adding edges between every two vertices $u$ and $v$ of $G$ with $d_G(u,v)=2$.
Araki \cite{AR2014} has studied the square of graphs and has proven that there exists a tree $T$ such that there are no two completely independent spanning trees in $T^2$ and that in the square of every $2$-connected graph, there are two completely independent spanning trees. Moreover, the family of trees such there are no two completely independent spanning trees in $T^2$ has been determined. We begin this section by proving that there exist two $(0,1)$-disjoint spanning trees in the square of every graph.

\begin{prop}
Let $G$ be graph.
There exist two $(0,1)$-disjoint spanning trees in $G^2$.
\end{prop}
\begin{proof}
Let $T$ be a spanning tree of $G$ let $V_1$ and $V_2$ be a bipartition of $T$.
The sets $V_1$ and $V_2$ form an $1$-CIST-partition of $G^2$ since both $G[V_1]$ and $G[V_2]$ are connected in $G^2$ and since $B(V_1,V_2)$ is a connected graph (which can be a tree in the case $G$ is a tree).
Thus, by Theorem \ref{theopart1}, there exist two $(0,1)$-disjoint spanning trees in $G^2$.
\end{proof}
We finish the section by determining which square of graph contains two completely independent spanning trees (the case of trees has already been treated \cite{AR2014}).

\begin{prop}
Let $G$ be a connected graph which is not a tree. There exist two completely independent spanning trees in $G^2$.
\end{prop}
\begin{proof}
Since $G$ is not a tree, there exists an induced cycle $C$ in $G$. Let $u$ be a vertex of $C$ which has a neighbor $v$ not belonging to $C$. If such vertex does not exist, then $G$ is cycle and $G^2$ contains two completely independent spanning trees \cite{AR2014}. Let $uw$ be an edge of $C$, let $T$ be a spanning tree of $G-uw$ and let $V_1$ and $V_2$ be a bipartition of $T$. Remark that  both $G[V_1]$ and $G[V_2]$ are connected and that every edge of $T$ belongs to $B(V_1,V_2)$. Our goal is to prove that there is one more edge in $B(V_1,V_2)$, i.e, that $B(V_1,V_2)$ is connected and is not a tree. First, if $C$ is of even length, then $u\in V_1$ and $w\in V_2$ (or $u\in V_2$ and $w\in V_1$, by symmetry) and $uw\in E(B(V_1,V_2))\setminus E(T)$. Second if $C$ is of odd length, then $u\in V_1$, $w\in V_1$ and $v\in V_2$ (or $v\in V_1$, $u\in V_2$ and $w\in V_2$, by symmetry) and $vw\in E(B(V_1,V_2))\setminus E(T)$. Thus, by Theorem \ref{theopart1}, there exist two completely independent spanning trees in $G^2$.
\end{proof}
Note that the square of a star (a tree of diameter at most 2) is a clique and can contain an arbitrary large number of completely independent spanning trees (this number depends on the degree of the central vertex). Thus, it could be interesting to determine which square of graph contains $k$ completely independent spanning trees for $k>2$.

\subsection{$k$-connected interval graph}
We begin by recalling the definition of a path-decomposition of a graph $G$.
\begin{de}
Let $G$ be a graph. A sequence of subsets $X_1,\ldots, X_\ell$ of vertices of $G$ is a path-decomposition of $G$ if the two following properties are satisfied:
\begin{enumerate}
\item[i)] for each edge $e$ of $G$, there exists an integer $i$ such that both extremities of $e$ belong to the subset $X_i$; 
\item[ii)] for every three integers $1\le i \le j \le k\le \ell$, $X_i \cap X_k \subseteq X_j$.
\end{enumerate}
\end{de}
An \emph{interval graph} is the intersection graph of a family of intervals of the real line. We recall that an interval graph has a path-decomposition $X_1,\ldots, X_\ell$ for which each $X_i$, $1\le i\le \ell$, forms a maximal clique in $G$. 
We also recall that for a $k$-connected interval graph $G$ with path-decomposition $X_1,\ldots, X_\ell$, we have $|X_i\cap X_{i+1}|\ge k$, for every integer $i$, $1\le i<\ell$ (otherwise $X_i\cap X_{i+1}$ would be a cut set of order less than $k$). The following property is true for $k$-connected interval graphs.
\begin{theo}
Let $k\ge 2$ be an integer.
Every $k$-connected interval graph $G$ satisfies $d_c(G)\ge k$. 
\end{theo}
\begin{proof}Let $G$ be a $k$-connected interval graph.
Let $X_1,\ldots, X_{\ell}$ be a path-decomposition of $G$, for which every $X_i$ forms a maximal clique.
If $\ell=1$, then $G$ is a $k$-connected complete graph, i.e., $G=K_n$ for $n\ge k$. Thus $G$ satisfies $d_c(G)=n\ge k$.
Hence, suppose $\ell\ge2$.
Our goal is to construct disjoint connected dominating sets $D_1,\ldots, D_k$ by setting $D_i=\cup_{1\le j\le \ell-1} \{x_i^j\}$, for $1\le i\le k$.

By hypothesis, $|X_1\cap X_2|\ge k$, and there exist $k$ different vertices $x_1^1,\ldots, x_k^1 \in X_1\cap X_2$ forming $k$ disjoint connected dominating sets on the graph $G[X_1\cup X_2]$. We set $D_i=\{x_i^1\}$, for every integer $i$, $1\le i\le k$.
Suppose $\ell>n$ and that, by induction, that we have already determined $x_i^{1},\ldots, x_i^{n-1}$, for every integer $i$, $1\le i\le k$ and that $D_1,\ldots, D_k$ are disjoint connected dominating sets on the graph $G[X_1\cup \ldots \cup X_{n}]$, for $D_i=\cup_{1\le j\le n-1} \{x_i^j\}$. Now our goal is to construct disjoint connected dominating sets on the graph $G[X_1\cup \ldots \cup X_{n+1}]$.

Let $i$ be an integer, $1\le i\le k$. If $X_{n}\cap X_{n+1}\cap \{x_i^{j}\}\neq\emptyset $, then we set $x_i^{n}= x_i^{n-1    }$, otherwise we set to $x_i^{n}$ a vertex not in $(X_{n}\cap X_{n+1}\cap\{x_1^{n},\ldots,x_k^{n}\})\setminus \{x_i^{n-1}\} $.
Such a vertex exists since otherwise it would imply that $|X_{n}\cap X_{n+1}|<k$.
Finally, the sets $D_1=\cup_{1\le j\le n} \{x_1^j\},\ldots,D_k=\cup_{1\le j\le n} \{x_k^j\}$ are disjoint connected dominating sets on the graph $G[X_1\cup \ldots \cup X_{n+1}]$. Consequently, by induction, we can construct $k$ disjoint connected dominating sets on the graph $G$.
\end{proof}
The previous theorem can not be generalized to chordal graphs since there exist $k$-connected chordal graphs, for $k\ge 2$, which do not contain two disjoint connected dominating sets \cite{FE2011}.
\subsection{Complete graphs}
By $\text{dst}_{i,j}(G)$ we denote the maximum number of $(i,j)$-disjoint spanning trees in $G$.
Remark that there are $n$ disjoint connected dominating sets in $K_n$ and that there are $\lfloor n/2 \rfloor$ completely independent spanning trees in $K_n$ \cite{MA2015}.
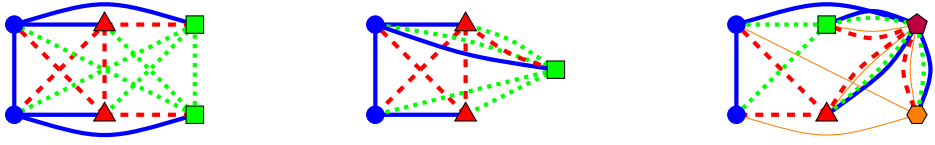
\begin{figure}[t]
\begin{center}
\begin{tikzpicture}[scale=1.2]
\draw[ultra thick,color=blue] (0,0) -- (0,1);
\draw[ultra thick,color=blue] (0,0) -- (1,0);
\draw[ultra thick,color=blue] (0,1) -- (1,1);
\draw[ultra thick,color=blue] (0,0) .. controls (1,-0.3) .. (2,0);
\draw[ultra thick,color=blue] (0,1) .. controls (1,1.3) .. (2,1);
\draw[ultra thick, style=dashed,color=red] (1,0) -- (1,1);
\draw[ultra thick, style=dashed,color=red] (0,0) -- (1,1);
\draw[ultra thick, style=dashed,color=red] (1,0) -- (0,1);
\draw[ultra thick, style=dashed,color=red] (1,0) -- (2,0);
\draw[ultra thick, style=dashed,color=red] (1,1) -- (2,1);
\draw[ultra thick, style=dotted,color=green] (2,0) -- (2,1);
\draw[ultra thick, style=dotted,color=green] (2,0) -- (1,1);
\draw[ultra thick, style=dotted,color=green] (2,0) -- (0,1);
\draw[ultra thick, style=dotted,color=green] (2,1) -- (1,0);
\draw[ultra thick, style=dotted,color=green] (2,1) -- (0,0);
\node at (0,0) [circle,draw=blue,fill=blue,scale=0.7]{};
\node at (1,0) [regular polygon, regular polygon sides=3,draw=black,fill=red,scale=0.5]{};
\node at (2,0) [regular polygon, regular polygon sides=4,draw=black,fill=green,scale=0.7]{};
\node at (0,1) [circle,draw=blue,fill=blue,scale=0.7]{};
\node at (1,1) [regular polygon, regular polygon sides=3,draw=black,fill=red,scale=0.5]{};
\node at (2,1) [regular polygon, regular polygon sides=4,draw=black,fill=green,scale=0.7]{};

\draw[ultra thick,color=blue] (0+4,0) -- (0+4,1);
\draw[ultra thick,color=blue] (0+4,0) -- (1+4,0);
\draw[ultra thick,color=blue] (0+4,1) -- (1+4,1);
\draw[ultra thick, style=dashed,color=red] (1+4,0) -- (1+4,1);
\draw[ultra thick, style=dashed,color=red] (0+4,0) -- (1+4,1);
\draw[ultra thick, style=dashed,color=red] (1+4,0) -- (0+4,1);
\draw[ultra thick, style=dotted,color=green] (2+4,0.5) -- (0+4,0);
\draw[ultra thick, style=dotted,color=green] (2+4,0.5) -- (1+4,0);
\draw[ultra thick, style=dotted,color=green] (2+4,0.5) .. controls (1+4,0.85).. (0+4,1);
\draw[ultra thick, style=dotted,color=green] (2+4,0.5) .. controls (1.5+4,0.85).. (1+4,1);
\draw[ultra thick,color=blue] (2+4,0.5) .. controls (1+4,0.65).. (0+4,1);
\draw[ultra thick, style=dashed,color=red] (2+4,0.5) .. controls (1.5+4,0.65).. (1+4,1);
\node at (0+4,0) [circle,draw=blue,fill=blue,scale=0.7]{};
\node at (1+4,0) [regular polygon, regular polygon sides=3,draw=black,fill=red,scale=0.5]{};
\node at (2+4,0.5) [regular polygon, regular polygon sides=4,draw=black,fill=green,scale=0.7]{};
\node at (0+4,1) [circle,draw=blue,fill=blue,scale=0.7]{};
\node at (1+4,1) [regular polygon, regular polygon sides=3,draw=black,fill=red,scale=0.5]{};

\draw[ultra thick,color=blue] (0+8,0) -- (0+8,1);
\draw[ultra thick,color=blue] (2+8,1) .. controls (1+8,1.3) .. (0+8,1);
\draw[ultra thick, style=dashed,color=red] (1+8,0) -- (0+8,0);
\draw[ultra thick, style=dashed,color=red] (1+8,0) -- (0+8,1);
\draw[ultra thick, style=dotted,color=green] (1+8,1) -- (0+8,0);
\draw[ultra thick, style=dotted,color=green] (1+8,1) -- (0+8,1);
\draw[color=orange] (2+8,0) .. controls (1+8,-0.3) .. (0+8,0);
\draw[color=orange] (2+8,0) -- (0+8,1);

\draw[ultra thick,color=blue] (2+8,1) .. controls (1.5+8,1.2).. (1+8,1);
\draw[ultra thick, style=dotted,color=green] (2+8,1) .. controls (1.5+8,1.1).. (1+8,1);
\draw[color=orange] (2+8,1) .. controls (1.5+8,0.9).. (1+8,1);
\draw[ultra thick, style=dashed,color=red] (2+8,1) .. controls (1.5+8,0.8).. (1+8,1);
\draw[ultra thick,color=blue] (2+8,1) .. controls (2.2+8,0.5).. (2+8,0);
\draw[ultra thick, style=dotted,color=green] (2+8,1) .. controls (2.1+8,0.5).. (2+8,0);
\draw[color=orange] (2+8,1) .. controls (1.9+8,0.5).. (2+8,0);
\draw[ultra thick, style=dashed,color=red] (2+8,1) .. controls (1.8+8,0.5).. (2+8,0);
\draw[ultra thick,color=blue] (2+8,1) .. controls (1.7+8,0.5).. (1+8,0);
\draw[ultra thick, style=dotted,color=green] (2+8,1) .. controls (1.6+8,0.5).. (1+8,0);
\draw[color=orange] (2+8,1) .. controls (1.4+8,0.5).. (1+8,0);
\draw[ultra thick, style=dashed,color=red] (2+8,1) .. controls (1.3+8,0.5).. (1+8,0);

\node at (0+8,0) [circle,draw=blue,fill=blue,scale=0.7]{};
\node at (1+8,0) [regular polygon, regular polygon sides=3,draw=black,fill=red,scale=0.5]{};
\node at (2+8,0) [regular polygon, regular polygon sides=6,draw=black,fill=orange,scale=0.7]{};
\node at (0+8,1) [circle,draw=blue,fill=blue,scale=0.7]{};
\node at (1+8,1) [regular polygon, regular polygon sides=4,draw=black,fill=green,scale=0.7]{};
\node at (2+8,1)  [regular polygon, regular polygon sides=5,draw=black,fill=purple,scale=0.7]{};
\end{tikzpicture}
\end{center}
\caption{Three completely independent spanning trees of $K_6$ (on the left), three $(0,2)$-disjoint spanning trees of $K_5$ (on the middle) and four $(1,3)$-disjoint spanning trees of $K_6$ (on the right) (thin line: edge of $T_4$; hexagon: $I(T_4)$; pentagon: $I(T_1)\cap I(T_2)\cap I(T_3)\cap I(T_4)$).}
\label{descliques}
\end{figure}

We give the following intermediate result about $(0,\ell)$-disjoint spanning trees.
\begin{prop}
Let $n$ be an integer. We have $\text{dst}_{0,\ell}(K_n)=\lfloor n/2\rfloor+ \min(\lfloor \ell/(n-1) +1_{\text{odd}}(n) /2 \rfloor, \lceil n/2 \rceil)$, where $1_{\text{odd}}(n)=1$ if $n$ is odd and $0$ otherwise.
\end{prop}
\begin{proof}
First, suppose $n$ is even. Let $i=\lfloor \ell/(n-1) \rfloor$. We begin by proving that $\text{dst}_{0,\ell}(K_n)< n/2 + i+1$ for $0\le i< n/2$. Suppose that there are $n/2+i+1$ $(0,\ell)$-disjoint spanning trees. 
By Corollary \ref{prop444}, we have $|E(K_n)|\ge (n/2+i+1)(n-1)-\ell$. Observe that $(n/2+i+1)(n-1)-\ell=n(n+1)/2+i(n-1)-1-\ell$. Since $|E(K_n)|= n (n-1)/2$, we have $\ell\ge  i(n-1)+n-1=(i+1)(n-1)$, contradicting the definition of $i$.
We are going to prove that we can construct $n/2+i$ $(0,\ell)$-disjoint spanning trees in $K_n$, for $ 1\le i\le n/2$.
We construct two kinds of spanning trees. First, we construct $2i$ spanning trees $T_1,\ldots, T_{2i}$ which are spanning stars. Second, we construct $n/2-i$ spanning trees in $K_n$ each tree with two inner vertices, as in \cite{MA2015} (with disjoint inner vertices). The left part of Figure \ref{descliques} illustrates this construction for $K_6$. There are $2i (2i-1)/2$ common edges between the spanning stars and $2i (n/2-i)$ common edges between the inner vertices of $T_1,\ldots, T_{2i}$ and the remaining vertices. Thus, there are $i(n-1)$ common edges and by definition $\ell\ge i(n-1)$.

Second, suppose $n$ is odd. Let $i=\lfloor (\ell/(n-1)+1/2 \rfloor$. We begin by proving that $\text{dst}_{0,\ell}(K_n)< (n-1)/2 + i+1$ for $0\le i< (n+1)/2$. Suppose that there are $(n-1)/2+i+1$ $(0,\ell)$-disjoint spanning trees. By Corollary \ref{prop444}, we have $|E(K_n)|\ge ((n-1)/2+i+1)(n-1)-\ell$. Observe that $((n-1)/2+i+1)(n-1)-\ell=n(n-1)/2+i(n-1)+(n-1)/2-\ell$. Since $|E(K_n)|= n (n-1)/2$, we have $\ell\ge  i(n-1)+(n-1)/2$, contradicting the definition of $i$.
We begin by constructing $(n-1)/2+i$ $(0,\ell)$-disjoint spanning trees in $K_n$, for $1\le i\le (n+1)/2$. We construct two kinds of spanning trees. First, we construct $2i-1$ spanning trees $T_1,\ldots, T_{2i-1}$ which are spanning stars. Second, we construct $(n-1)/2-i+1$ spanning trees in $K_n$, each tree having two inner vertices, following the construction described in \cite{MA2015} (with disjoint inner vertices). There are $(2i-1) (2i-2)/2$ common edges between the spanning stars and $(2i-1) ((n-1)/2-i)$ common edges between the inner vertices of $T_1,\ldots, T_{2i-1}$ and the remaining vertices. Thus, there are $(2i-1)(i-1)+(2i-1)((n-1)/2-i+1)= i(n-1)-(n-1)/2$ common edges and by definition $\ell\ge i(n-1)-(n-1)/2$.
\end{proof}
The middle part of Figure \ref{descliques} depicts three $(0,2)$-disjoint spanning trees in $K_5$.
\begin{prop}
Let $n$ be a positive integer. For $1\le \ell<(n-1)$, we have $\text{dst}_{1,\ell}(K_n)\le \lfloor n/2 \rfloor+ \lfloor \ell/2-1_{\text{even}}(n) /2 \rfloor$, where $1_{\text{even}}(n)=1$ if $n$ is even, and $0$ otherwise. Moreover, if $\ell\ge (n-1)$, then $\text{dst}_{1,\ell}(K_n)$ is not finite.
\end{prop}
\begin{proof}
Observe that a connected dominating set of $K_n$ can contain only one vertex. Thus, if $\ell\ge (n-1)$, then $\text{dst}_{1,\ell}(K_n)$ is not finite.

First suppose $n$ is even. Let $i= \lfloor \ell/2 -1/2\rfloor$. We prove that we can construct $n/2 + i$ $(1,\ell)$-disjoint spanning trees in $K_n$, for $1\le i\le n/2-1$. 
Let $B$ be an induced $K_{n-2(i+1)}$ of $G=K_n$.
We begin by creating $n/2-i-1$ completely independent spanning trees $T_1,\ldots, T_{n/2-i-1}$ in $B$, as in \cite{MA2015}.
Let $u$ be a vertex of $G- B$. We are going to extend these trees in order they span the whole graph $G$. To each tree $T_k$, $1 \le k\le n/2-i-1$, we add the vertex $u$ and an edge incident to $u$ and to a vertex of $I(T_k)\cap B$. We finally add to $T_k$ the edges incident to $u$ and to every vertex of $G-B$.
We now construct the remaining trees $T'_1,\ldots, T'_{2i+1}$ as follows. Each tree $T'_k$, $1 \le k\le 2i+1$, has two inner vertices: $u$ and a vertex $u_k$ of $G-(B\cup \{u\})$ different for each tree. Each tree $T'_{k}$ also contains the edges incident to $u$ and to every vertex of $G-B$ and the edges incident to $u_k$ and to every vertex of $B$.
It is easy to verify that the trees $T_1,\ldots, T_{n/2-i-1},T'_1,\ldots, T'_{2i+1}$ have only one common vertex (the vertex $u$) and $2i+1=\ell$ common edges (the edges incident to $u$ in $G-B$).

Second, suppose $n$ is odd. Let $i= \lfloor \ell/2 \rfloor$. We prove that we can construct $(n-1)/2 + i$ $(1,\ell)$-disjoint spanning trees in $K_n$, for $1\le i\le (n-1)/2$. 
Let $B$ be an induced $K_{n-2i-1}$ of $K_n$.
We begin by creating $(n-1)/2-i$ completely spanning trees $T_1,\ldots, T_{(n-1)/2-i}$ in $B$, as in \cite{MA2015} and extend them to the whole graph $G$ as for the case $n$ even.
We construct the trees $T'_1,\ldots, T'_{2i}$ similarly as in the case $n$ even.
It is easy to verify that the trees $T_1,\ldots, T_{(n-1)/2-i},T'_1,\ldots, T'_{2i}$ have only one common vertex (the vertex $u$) and $2i$ common edges (the edges incident to $u$ in $G-B$).
\end{proof}
The right part of Figure \ref{descliques} depicts four $(1,3)$-disjoint spanning trees in $K_6$.
Note that we can obtain a lower bound on the number of $(1,\ell)$-disjoint trees in $K_n$ by using Proposition \ref{prop44}.
However, in this case, we do not obtain a tight bound. Moreover, Proposition \ref{prop44} implies that $\text{dst}_{i,0}(K_n)=\lfloor n/2 \rfloor$ for every positive integer $i$.
\subsection{Cylinders}
Let $n_1$ and $n_2$ be positive integers with $n_1\ge 3$ and $n_2\ge 3$. 
Let $V(C(n_1,n_2))=\{(i,j)|\ 0\le i<n_1, 0\le j<n_2 \}$ and $E(C(n_1,n_2))=\{(i,j)\ (i',j'')|\ i=i', j=j'\pm 1 \lor j=j', |i-i'|=1\pmod{n_1} \}$. 

\begin{theo}
There exist two $(0,n_1- 2)$-disjoint spanning trees in the cylinder $C(n_1,n_2)$.

\end{theo}
\begin{proof}
We describe these two trees by giving their edge sets:

$E(T_1)=\{(0,j) (0,j+1)|\ j\in\{0,\ldots, n_2-2\} \}\cup \{(i,j) (i+1,j)|\ i\in \{0,\ldots,n_1 -2\},\  j\in\{0,2,4,\ldots, 2 \lfloor (n_2-1)/2 \rfloor \} \} \cup \{(i,j) (i,j+1)|\  i\in \{1,\ldots,n_1 -2\},\ j\in\{0,2,4,\ldots, 2 \lfloor (n_2-1)/2 \rfloor \} \}\cup \{(0,j) (n_1-1,j)|\ j\in\{1,3,5,\ldots, 2 \lfloor (n_2-2)/2 \rfloor+1 \} \}$.

$E(T_2)=\{(n_1-1,j) (n_1-1,j+1)|\ j\in\{0,\ldots, n_2-2\}\}\cup \{(i,j) (i+1,j)|\ i\in \{0,\ldots ,n_1-2\},\ j\in\{1,3,5,\ldots, 2 \lfloor (n_2-2)/2 \rfloor+1 \}  \} \cup \{(i,j) (i,j+1)|\  i\in \{1,\ldots,n_1 -2\},\ j\in\{1,3,5,\ldots, 2 \lfloor (n_2-2)/2 \rfloor+1 \} \} \cup \{(0,j) (n_1-1,j)|\ j\in\{1,3,5,\ldots, 2 \lfloor (n_2-1)/2 \rfloor \} \}\cup \{(i,0) (i,1)|\  i\in \{1,\ldots,n_1 -2\}\}$.
\end{proof}
\begin{figure}[t]
\begin{center}
\begin{tikzpicture}
\draw[ultra thick,color=blue] (0,0) -- (0,4);
\draw[ultra thick,color=blue] (2,0) -- (2,4);
\draw[ultra thick,color=blue] (4,0) -- (4,4);
\draw[ultra thick,color=blue] (1,0) .. controls (0.5,2) .. (1,4);
\draw[ultra thick,color=blue] (3,0) .. controls (2.5,2) .. (3,4);
\draw[ultra thick,color=blue] (5,0) .. controls (4.5,2) .. (5,4);
\draw[ultra thick,color=blue] (0,0) -- (5,0);

\draw[ultra thick,color=blue] (2,1) -- (3,1);
\draw[ultra thick,color=blue] (2,2) -- (3,2);
\draw[ultra thick,color=blue] (2,3) -- (3,3);

\draw[ultra thick,color=blue] (4,1) -- (5,1);
\draw[ultra thick,color=blue] (4,2) -- (5,2);
\draw[ultra thick,color=blue] (4,3) -- (5,3);

\draw[ultra thick, style=dashed,color=red] (1,0) -- (1,4);
\draw[ultra thick, style=dashed,color=red] (3,0) -- (3,4);
\draw[ultra thick, style=dashed,color=red] (5,0) -- (5,4);
%\draw[ultra thick, style=dashed,color=red] (2,1) -- (2,0);
\draw[ultra thick, style=dashed,color=red] (0,0) .. controls (-0.5,2) .. (0,4);
\draw[ultra thick, style=dashed,color=red] (2,0) .. controls (1.5,2) .. (2,4);
\draw[ultra thick, style=dashed,color=red] (4,0) .. controls (3.5,2) .. (4,4);
\draw[ultra thick, style=dashed,color=red] (0,4) -- (5,4);

\draw[ultra thick,color=blue] (0,1) .. controls (0.5,1.1).. (1,1);
\draw[ultra thick,color=blue] (0,2) .. controls (0.5,2.1).. (1,2);
\draw[ultra thick, color=blue] (0,3) .. controls (0.5,3.1).. (1,3);
\draw[ultra thick, style=dashed,color=red] (0,1) .. controls (0.5,0.9).. (1,1);
\draw[ultra thick, style=dashed,color=red] (0,2) .. controls (0.5,1.9).. (1,2);
\draw[ultra thick, style=dashed,color=red] (0,3) .. controls (0.5,2.9).. (1,3);

\draw[ultra thick, style=dashed,color=red] (1,1) -- (2,1);
\draw[ultra thick, style=dashed,color=red] (1,2) -- (2,2);
\draw[ultra thick, style=dashed,color=red] (1,3) -- (2,3);

\draw[ultra thick, style=dashed,color=red] (3,1) -- (4,1);
\draw[ultra thick, style=dashed,color=red] (3,2) -- (4,2);
\draw[ultra thick, style=dashed,color=red] (3,3) -- (4,3);

\node at (0,0) [circle,draw=black,fill=blue,scale=0.7]{};
\node at (0,1) [circle,draw=black,fill=blue,scale=0.7]{};
\node at (0,2) [circle,draw=black,fill=blue,scale=0.7]{};
\node at (0,3) [circle,draw=black,fill=blue,scale=0.7]{};
\node at (1,0) [circle,draw=black,fill=blue,scale=0.7]{};
\node at (2,0) [circle,draw=black,fill=blue,scale=0.7]{};
\node at (2,1) [circle,draw=black,fill=blue,scale=0.7]{};
\node at (2,2) [circle,draw=black,fill=blue,scale=0.7]{};
\node at (2,3) [circle,draw=black,fill=blue,scale=0.7]{};
\node at (3,0) [circle,draw=black,fill=blue,scale=0.7]{};
\node at (4,0) [circle,draw=black,fill=blue,scale=0.7]{};
\node at (4,1) [circle,draw=black,fill=blue,scale=0.7]{};
\node at (4,2) [circle,draw=black,fill=blue,scale=0.7]{};
\node at (4,3) [circle,draw=black,fill=blue,scale=0.7]{};
\node at (5,0) [circle,draw=black,fill=blue,scale=0.7]{};
\node at (0,4) [regular polygon, regular polygon sides=3,draw=black,fill=red,scale=0.5]{};
\node at (1,1) [regular polygon, regular polygon sides=3,draw=black,fill=red,scale=0.5]{};
\node at (1,2) [regular polygon, regular polygon sides=3,draw=black,fill=red,scale=0.5]{};
\node at (1,3) [regular polygon, regular polygon sides=3,draw=black,fill=red,scale=0.5]{};
\node at (1,4) [regular polygon, regular polygon sides=3,draw=black,fill=red,scale=0.5]{};
\node at (2,4) [regular polygon, regular polygon sides=3,draw=black,fill=red,scale=0.5]{};
\node at (3,1) [regular polygon, regular polygon sides=3,draw=black,fill=red,scale=0.5]{};
\node at (3,2) [regular polygon, regular polygon sides=3,draw=black,fill=red,scale=0.5]{};
\node at (3,3) [regular polygon, regular polygon sides=3,draw=black,fill=red,scale=0.5]{};
\node at (3,4) [regular polygon, regular polygon sides=3,draw=black,fill=red,scale=0.5]{};
\node at (4,4) [regular polygon, regular polygon sides=3,draw=black,fill=red,scale=0.5]{};
\node at (5,1) [regular polygon, regular polygon sides=3,draw=black,fill=red,scale=0.5]{};
\node at (5,2) [regular polygon, regular polygon sides=3,draw=black,fill=red,scale=0.5]{};
\node at (5,3) [regular polygon, regular polygon sides=3,draw=black,fill=red,scale=0.5]{};
\node at (5,4) [regular polygon, regular polygon sides=3,draw=black,fill=red,scale=0.5]{};
\end{tikzpicture}
\end{center}
\caption{Two $(0,3)$-disjoint spanning trees of $C(5, 6)$.}
\label{cone}
\end{figure}
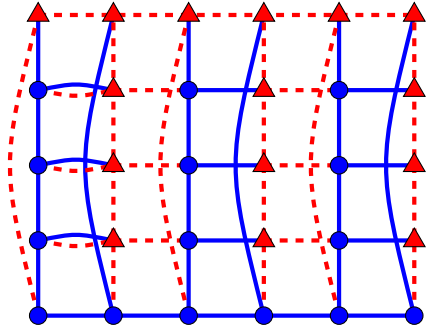
Observe that $C(n_1,n_2)$ contains $2 n_1 n_2 -n_1$ edges. Hence, by Corollary \ref{prop444}, we can conclude that there does not exist two $(0,m)$-disjoint spanning trees in $C(n_1,n_2)$, for $m< n_1-2$.

\subsection{square grids}
Let $n_1$ and $n_2$ be positive integers with $n_1\ge 3$ and $n_2\ge 3$. 
Let $V(G(n_1,n_2))=\{(i,j)|\ 0\le i<n_1, 0\le j<n_2 \}$ and $E(G(n_1,n_2))=\{(i,j)\ (i',j'')|\ i=i'\pm1, j=j' \lor i=i', i= j'\pm1 \}$. 

In two papers \cite{FU2003,LI2010}, the trees with a maximum number of leaves in $G(n_1,n_2)$ have been determined. In particular, Fujie \cite{FU2003} has shown that a spanning tree of $G(n_1,n_2)$ has at least $\lceil n_1 n_2 /3\rceil$ inner vertices.
Hartnell and Rall \cite{HAR2001} have proven that there do not exist two disjoint connected dominating sets in $G(n_1,n_2)$, except if $n_1\le 2$ or $n_2\le 2$. However, this is not the case for $1$-rooted connected dominating set.
We finish this paper by giving a construction of two $1$-rooted connected dominating sets in $G(n_1,n_2)$ for $n_1\ge n_1$ and $n_2\ge 3$. In Figure \ref{CDSgrille}, we exhibit two trees induced by two $1$-rooted connected dominating sets in $G(7,13)$. In this example, we have minimized the number of common edges.
\begin{theo}
There exist two $1$-rooted connected dominating sets in the grid $G(n_1,n_2)$, for every $n_1\ge 3$ and $n_2\ge 3$.
\end{theo}
\begin{figure}
\begin{center}
\begin{tikzpicture}[scale=1]

\draw[ultra thick,color=blue] (0,0) -- (12,0);
\draw[ultra thick, style=dashed,color=red] (0,1) -- (12,1);
\draw[ultra thick,color=blue] (1,2) -- (9,2);
\draw[ultra thick, style=dashed,color=red] (2,3) -- (8,3);
\draw[ultra thick,color=blue] (3,4) -- (9,4);
\draw[ultra thick, style=dashed,color=red] (2,5) -- (10,5);
\draw[ultra thick,color=blue] (1,6) -- (11,6);
\draw[ultra thick, style=dashed,color=red] (0,1) -- (0,6);
\draw[ultra thick,color=blue] (1,2) -- (1,6);
\draw[ultra thick, style=dashed,color=red] (2,3) -- (2,5);
\draw[ultra thick,color=blue] (9,2) -- (9,4);
\draw[ultra thick, style=dashed,color=red] (10,1) -- (10,5);
\draw[ultra thick,color=blue] (11,1) -- (11,6);
\draw[ultra thick, style=dashed,color=red] (12,1) -- (12,6);
\draw[ultra thick, style=dashed,color=red]  (0,0) .. controls (-0.1,0.5) .. (0,1);
\draw[ultra thick, style=dashed,color=red]  (1,0) .. controls (0.9,0.5) .. (1,1);
\draw[ultra thick, style=dashed,color=red]  (2,0) .. controls (1.9,0.5) ..  (2,1);
\draw[ultra thick, style=dashed,color=red] (3,0) .. controls (2.9,0.5) .. (3,1);
\draw[ultra thick, style=dashed,color=red] (4,0) .. controls (3.9,0.5) .. (4,1);
\draw[ultra thick, style=dashed,color=red] (5,0) .. controls (4.9,0.5) .. (5,1);
\draw[ultra thick, style=dashed,color=red] (6,0) .. controls (5.9,0.5) .. (6,1);
\draw[ultra thick, style=dashed,color=red] (7,0) .. controls (6.9,0.5) .. (7,1);
\draw[ultra thick, style=dashed,color=red] (8,0) .. controls (7.9,0.5) .. (8,1);
\draw[ultra thick, style=dashed,color=red] (9,0) .. controls (8.9,0.5) .. (9,1);
\draw[ultra thick, style=dashed,color=red] (10,0) .. controls (9.9,0.5) ..  (10,1);
\draw[ultra thick, style=dashed,color=red] (11,0) .. controls (10.9,0.5) ..  (11,1);
\draw[ultra thick, style=dashed,color=red]  (12,0) .. controls (11.9,0.5) ..  (12,1);
\draw[ultra thick,color=blue]  (0,0) .. controls (0.1,0.5) .. (0,1);
\draw[ultra thick,color=blue]  (1,0) .. controls (1.1,0.5) .. (1,1);
\draw[ultra thick,color=blue] (2,0) .. controls (2.1,0.5) ..  (2,1);
\draw[ultra thick,color=blue] (3,0) .. controls (3.1,0.5) .. (3,1);
\draw[ultra thick,color=blue] (4,0) .. controls (4.1,0.5) .. (4,1);
\draw[ultra thick,color=blue] (5,0) .. controls (5.1,0.5) .. (5,1);
\draw[ultra thick,color=blue] (6,0) .. controls (6.1,0.5) .. (6,1);
\draw[ultra thick,color=blue] (7,0) .. controls (7.1,0.5) .. (7,1);
\draw[ultra thick,color=blue] (8,0) .. controls (8.1,0.5) .. (8,1);
\draw[ultra thick,color=blue] (9,0) .. controls (9.1,0.5) .. (9,1);
\draw[ultra thick,color=blue] (10,0) .. controls (10.1,0.5) ..  (10,1);
\draw[ultra thick,color=blue] (11,0) .. controls (11.1,0.5) ..  (11,1);
\draw[ultra thick,color=blue] (12,0) .. controls (12.1,0.5) ..  (12,1);

\draw[ultra thick, style=dashed,color=red] (1,1) -- (1,2);
\draw[ultra thick, style=dashed,color=red] (2,1) -- (2,2);
\draw[ultra thick, style=dashed,color=red] (3,1) -- (3,2);
\draw[ultra thick, style=dashed,color=red] (4,1) -- (4,2);
\draw[ultra thick, style=dashed,color=red] (5,1) -- (5,2);
\draw[ultra thick, style=dashed,color=red] (6,1) -- (6,2);
\draw[ultra thick, style=dashed,color=red] (7,1) -- (7,2);
\draw[ultra thick, style=dashed,color=red] (8,1) -- (8,2);
\draw[ultra thick, style=dashed,color=red] (9,1) -- (9,2);
\draw[ultra thick,color=blue](2,2) -- (2,3);
\draw[ultra thick,color=blue] (3,2) -- (3,3);
\draw[ultra thick,color=blue] (4,2) -- (4,3);
\draw[ultra thick,color=blue] (5,2) -- (5,3);
\draw[ultra thick,color=blue] (6,2) -- (6,3);
\draw[ultra thick,color=blue] (7,2) -- (7,3);
\draw[ultra thick,color=blue] (8,2) -- (8,3);
\draw[ultra thick, style=dashed,color=red] (2,3) -- (2,4);
\draw[ultra thick, style=dashed,color=red] (3,3) -- (3,4);
\draw[ultra thick, style=dashed,color=red] (4,3) -- (4,4);
\draw[ultra thick, style=dashed,color=red] (5,3) -- (5,4);
\draw[ultra thick, style=dashed,color=red] (6,3) -- (6,4);
\draw[ultra thick, style=dashed,color=red] (7,3) -- (7,4);
\draw[ultra thick, style=dashed,color=red] (8,3) -- (8,4);
\draw[ultra thick,color=blue] (3,4) -- (3,5);
\draw[ultra thick,color=blue] (4,4) -- (4,5);
\draw[ultra thick,color=blue] (5,4) -- (5,5);
\draw[ultra thick,color=blue] (6,4) -- (6,5);
\draw[ultra thick,color=blue] (7,4) -- (7,5);
\draw[ultra thick,color=blue] (8,4) -- (8,5);
\draw[ultra thick,color=blue] (9,4) -- (9,5);
\draw[ultra thick, style=dashed,color=red] (2,5) .. controls (1.9,5.5) .. (2,6);
\draw[ultra thick,color=blue] (2,5) .. controls (2.1,5.5) .. (2,6);
\draw[ultra thick, style=dashed,color=red] (3,5) -- (3,6);
\draw[ultra thick, style=dashed,color=red] (4,5) -- (4,6);
\draw[ultra thick, style=dashed,color=red] (5,5) -- (5,6);
\draw[ultra thick, style=dashed,color=red] (6,5) -- (6,6);
\draw[ultra thick, style=dashed,color=red] (7,5) -- (7,6);
\draw[ultra thick, style=dashed,color=red] (8,5) -- (8,6);
\draw[ultra thick, style=dashed,color=red] (9,5) -- (9,6);
\draw[ultra thick,color=blue] (10,5) .. controls (10.1,5.5) .. (10,6);
\draw[ultra thick, style=dashed,color=red] (10,5) .. controls (9.9,5.5) .. (10,6);

\draw[ultra thick,color=blue] (0,2) -- (1,2);
\draw[ultra thick,color=blue] (0,3) -- (1,3);
\draw[ultra thick,color=blue] (0,4) -- (1,4);
\draw[ultra thick,color=blue]  (0,5) -- (1,5);
\draw[ultra thick, style=dashed,color=red] (0,6) .. controls (0.5,6.1) .. (1,6);
\draw[ultra thick,color=blue] (0,6) .. controls (0.5,5.9) .. (1,6);

\draw[ultra thick, style=dashed,color=red] (1,3) -- (2,3);
\draw[ultra thick, style=dashed,color=red] (1,4) -- (2,4);
\draw[ultra thick, style=dashed,color=red]  (1,5) -- (2,5);
\draw[ultra thick,color=blue] (2,4) -- (3,4);
\draw[ultra thick, style=dashed,color=red]  (8,3) -- (9,3);
\draw[ultra thick,color=blue] (9,2) -- (10,2);
\draw[ultra thick,color=blue] (9,3) -- (10,3);
\draw[ultra thick, style=dashed,color=red] (9,4) -- (10,4);
\draw[ultra thick, style=dashed,color=red] (10,2) -- (11,2);
\draw[ultra thick, style=dashed,color=red] (10,3) -- (11,3);
\draw[ultra thick,color=blue] (10,4) -- (11,4);
\draw[ultra thick, style=dashed,color=red] (10,5) -- (11,5);
\draw[ultra thick,color=blue] (11,2) -- (12,2);
\draw[ultra thick,color=blue] (11,3) -- (12,3);
\draw[ultra thick, style=dashed,color=red] (11,4) .. controls (11.5,4.1) .. (12,4);
\draw[ultra thick,color=blue] (11,4) .. controls (11.5,3.9) .. (12,4);

\draw[ultra thick,color=blue] (11,5) -- (12,5);
\draw[ultra thick, style=dashed,color=red] (11,6) .. controls (11.5,6.1) .. (12,6);
\draw[ultra thick,color=blue] (11,6) .. controls (11.5,5.9) .. (12,6);

\node at (0,0) [circle,draw=black,fill=blue,scale=0.7]{};
\node at (1,0) [circle,draw=black,fill=blue,scale=0.7]{};
\node at (2,0) [circle,draw=black,fill=blue,scale=0.7]{};
\node at (3,0) [circle,draw=black,fill=blue,scale=0.7]{};
\node at (4,0) [circle,draw=black,fill=blue,scale=0.7]{};
\node at (5,0) [circle,draw=black,fill=blue,scale=0.7]{};
\node at (6,0) [circle,draw=black,fill=blue,scale=0.7]{};
\node at (7,0) [circle,draw=black,fill=blue,scale=0.7]{};
\node at (8,0) [circle,draw=black,fill=blue,scale=0.7]{};
\node at (9,0) [circle,draw=black,fill=blue,scale=0.7]{};
\node at (10,0) [circle,draw=black,fill=blue,scale=0.7]{};
\node at (11,0) [circle,draw=black,fill=blue,scale=0.7]{};
\node at (12,0) [circle,draw=black,fill=blue,scale=0.7]{};
\node at (0,1) [regular polygon, regular polygon sides=3,draw=black,fill=red,scale=0.5]{};
\node at (1,1) [regular polygon, regular polygon sides=3,draw=black,fill=red,scale=0.5]{};
\node at (2,1) [regular polygon, regular polygon sides=3,draw=black,fill=red,scale=0.5]{};
\node at (3,1) [regular polygon, regular polygon sides=3,draw=black,fill=red,scale=0.5]{};
\node at (4,1) [regular polygon, regular polygon sides=3,draw=black,fill=red,scale=0.5]{};
\node at (5,1) [regular polygon, regular polygon sides=3,draw=black,fill=red,scale=0.5]{};
\node at (6,1) [regular polygon, regular polygon sides=3,draw=black,fill=red,scale=0.5]{};
\node at (7,1) [regular polygon, regular polygon sides=3,draw=black,fill=red,scale=0.5]{};
\node at (8,1) [regular polygon, regular polygon sides=3,draw=black,fill=red,scale=0.5]{};
\node at (9,1) [regular polygon, regular polygon sides=3,draw=black,fill=red,scale=0.5]{};
\node at (10,1) [regular polygon, regular polygon sides=3,draw=black,fill=red,scale=0.5]{};
\node at (11,1)  [regular polygon, regular polygon sides=4,draw=black,fill=purple,scale=0.7]{};
\node at (12,1) [regular polygon, regular polygon sides=3,draw=black,fill=red,scale=0.5]{};
\node at (0,2) [regular polygon, regular polygon sides=3,draw=black,fill=red,scale=0.5]{};
\node at (1,2) [circle,draw=black,fill=blue,scale=0.7]{};
\node at (2,2) [circle,draw=black,fill=blue,scale=0.7]{};
\node at (3,2) [circle,draw=black,fill=blue,scale=0.7]{};
\node at (4,2) [circle,draw=black,fill=blue,scale=0.7]{};
\node at (5,2) [circle,draw=black,fill=blue,scale=0.7]{};
\node at (6,2) [circle,draw=black,fill=blue,scale=0.7]{};
\node at (7,2) [circle,draw=black,fill=blue,scale=0.7]{};
\node at (8,2) [circle,draw=black,fill=blue,scale=0.7]{};
\node at (9,2) [circle,draw=black,fill=blue,scale=0.7]{};
\node at (10,2) [regular polygon, regular polygon sides=3,draw=black,fill=red,scale=0.5]{};
\node at (11,2) [circle,draw=black,fill=blue,scale=0.7]{};
\node at (12,2) [regular polygon, regular polygon sides=3,draw=black,fill=red,scale=0.5]{};
\node at (0,3) [regular polygon, regular polygon sides=3,draw=black,fill=red,scale=0.5]{};
\node at (1,3) [circle,draw=black,fill=blue,scale=0.7]{};
\node at (2,3) [regular polygon, regular polygon sides=3,draw=black,fill=red,scale=0.5]{};
\node at (3,3) [regular polygon, regular polygon sides=3,draw=black,fill=red,scale=0.5]{};
\node at (4,3) [regular polygon, regular polygon sides=3,draw=black,fill=red,scale=0.5]{};
\node at (5,3) [regular polygon, regular polygon sides=3,draw=black,fill=red,scale=0.5]{};
\node at (6,3) [regular polygon, regular polygon sides=3,draw=black,fill=red,scale=0.5]{};
\node at (7,3) [regular polygon, regular polygon sides=3,draw=black,fill=red,scale=0.5]{};
\node at (8,3) [regular polygon, regular polygon sides=3,draw=black,fill=red,scale=0.5]{};
\node at (9,3) [circle,draw=black,fill=blue,scale=0.7]{};
\node at (10,3) [regular polygon, regular polygon sides=3,draw=black,fill=red,scale=0.5]{};
\node at (11,3) [circle,draw=black,fill=blue,scale=0.7]{};
\node at (12,3) [regular polygon, regular polygon sides=3,draw=black,fill=red,scale=0.5]{};
\node at (0,4) [regular polygon, regular polygon sides=3,draw=black,fill=red,scale=0.5]{};
\node at (1,4) [circle,draw=black,fill=blue,scale=0.7]{};
\node at (2,4) [regular polygon, regular polygon sides=3,draw=black,fill=red,scale=0.5]{};
\node at (3,4) [circle,draw=black,fill=blue,scale=0.7]{};
\node at (4,4) [circle,draw=black,fill=blue,scale=0.7]{};
\node at (5,4) [circle,draw=black,fill=blue,scale=0.7]{};
\node at (6,4) [circle,draw=black,fill=blue,scale=0.7]{};
\node at (7,4) [circle,draw=black,fill=blue,scale=0.7]{};
\node at (8,4) [circle,draw=black,fill=blue,scale=0.7]{};
\node at (9,4) [circle,draw=black,fill=blue,scale=0.7]{};
\node at (10,4) [regular polygon, regular polygon sides=3,draw=black,fill=red,scale=0.5]{};
\node at (11,4) [circle,draw=black,fill=blue,scale=0.7]{};
\node at (12,4) [regular polygon, regular polygon sides=3,draw=black,fill=red,scale=0.5]{};
\node at (0,5) [regular polygon, regular polygon sides=3,draw=black,fill=red,scale=0.5]{};
\node at (1,5) [circle,draw=black,fill=blue,scale=0.7]{};
\node at (2,5) [regular polygon, regular polygon sides=3,draw=black,fill=red,scale=0.5]{};
\node at (2,5) [regular polygon, regular polygon sides=3,draw=black,fill=red,scale=0.5]{};
\node at (3,5) [regular polygon, regular polygon sides=3,draw=black,fill=red,scale=0.5]{};
\node at (4,5) [regular polygon, regular polygon sides=3,draw=black,fill=red,scale=0.5]{};
\node at (5,5) [regular polygon, regular polygon sides=3,draw=black,fill=red,scale=0.5]{};
\node at (6,5) [regular polygon, regular polygon sides=3,draw=black,fill=red,scale=0.5]{};
\node at (7,5) [regular polygon, regular polygon sides=3,draw=black,fill=red,scale=0.5]{};
\node at (8,5) [regular polygon, regular polygon sides=3,draw=black,fill=red,scale=0.5]{};
\node at (9,5) [regular polygon, regular polygon sides=3,draw=black,fill=red,scale=0.5]{};
\node at (10,5) [regular polygon, regular polygon sides=3,draw=black,fill=red,scale=0.5]{};
\node at (11,5) [circle,draw=black,fill=blue,scale=0.7]{};
\node at (12,5) [regular polygon, regular polygon sides=3,draw=black,fill=red,scale=0.5]{};
\node at (0,6) [regular polygon, regular polygon sides=3,draw=black,fill=red,scale=0.5]{};
\node at (1,6) [circle,draw=black,fill=blue,scale=0.7]{};
\node at (2,6) [circle,draw=black,fill=blue,scale=0.7]{};
\node at (3,6) [circle,draw=black,fill=blue,scale=0.7]{};
\node at (4,6) [circle,draw=black,fill=blue,scale=0.7]{};
\node at (5,6) [circle,draw=black,fill=blue,scale=0.7]{};
\node at (6,6) [circle,draw=black,fill=blue,scale=0.7]{};
\node at (7,6) [circle,draw=black,fill=blue,scale=0.7]{};
\node at (8,6) [circle,draw=black,fill=blue,scale=0.7]{};
\node at (9,6) [circle,draw=black,fill=blue,scale=0.7]{};
\node at (10,6) [circle,draw=black,fill=blue,scale=0.7]{};
\node at (11,6) [circle,draw=black,fill=blue,scale=0.7]{};
\node at (12,6) [regular polygon, regular polygon sides=3,draw=black,fill=red,scale=0.5]{};

\end{tikzpicture}
\end{center}
\caption{Two $(1,18)$-disjoint spanning trees of $G(7, 13)$.}
\label{CDSgrille}
\end{figure}
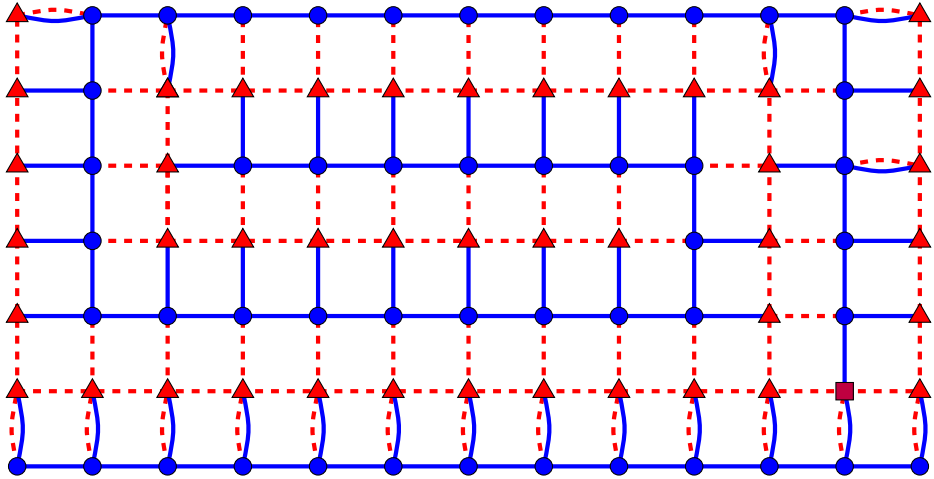
\begin{proof}
Suppose without loss of generality that $n_1\le n_2$. If $n_1=3$, then one can easily construct two $1$-rooted connected dominating sets by setting $D_1=\{(1,j)|j\in \{0,\ldots,n_2-1\}\}$ and by setting $D_2=V(G(3,n_2))\setminus D_1\cup \{(1,0)\}$.

Now suppose that $n_1\ge 4$. We construct $D_1$ as follows: \newline
$D_1=\{(0,i)|\ i\in\{ 0,\ldots,n_2-1\} \} \cup \{(n_1-1-2i,j)|\ i\in \{0,\ldots, \lfloor (n_1+1)/4 \rfloor \},\ 2i+1\le j\le n_2-2-2i \}$ $\cup \{(2+2i,j)|\ i\in \{0,\ldots, \lfloor (n_1-1)/4 \rfloor\},\ 2i+1 \le j\le n_2-4-2i \} \cup \{(i, 1+2j)|\ 2j+2\le i \le n_1-1-2j,\ j\in \{0,\ldots, \lfloor n_1/4 \rfloor \}\cup \{(i, n_2-2-2j)|\ 2j\le i\le n_1-1-2j,\ j\in \{0,\ldots, \lfloor (n_1+2)/4 \rfloor \}$.\newline
The set $D_2$ is $V(G(n_1,n_2))\setminus D_1 \cup \{(1,n_2-2)\}$. Note that $D_1\cap D_2=\{(1,n_2-2)\}$.

Figure \ref{CDSgrille} illustrates this construction, with circle vertices corresponding to $D_1$ and triangles to $D_2$ (the square vertex being both in $D_1$ and $D_2$).
\end{proof}

\section*{Open questions}

In this introducing paper about $(i,j)$-disjoint spanning trees, we tried to cover a large number of issues. However, there still remains a lot of interesting properties to be found about this notion. We finish this paper by giving some open questions:
\begin{enumerate}
\item Does there exist $k$ $( log_k(n), 0)$-disjoint spanning trees in every $k$-connected graph of order $n$?
\item Determine conditions in order to guarantee the existence of $k$ completely independent spanning trees, $k\ge 3$, in the square of graphs.
\item Determine conditions on chordal graphs in order to guarantee the existence of two disjoint connected dominating set.
\item Determine $\text{dst}_{i,j}(K_n)$ for the remaining cases.
\item Determine $\text{dst}_{i,j}$ for the complete $k$-partite graphs.
\item Determine the minimum number of common edges $m$ in order to have two $(1,m)$-disjoint spanning trees in the square grid.
\end{enumerate}

\end{document}